\def\isarxiv{1}

\documentclass[11pt]{article}
\usepackage[margin=1in]{geometry}
\usepackage{latexsym, amscd, amsfonts, mathrsfs, amsmath, amssymb, amsthm}
\usepackage{bm}
\usepackage{color}
\usepackage[noend]{algpseudocode}
\usepackage[colorlinks,citecolor=blue,linkcolor=blue,urlcolor=red,pagebackref]{hyperref}
\usepackage[capitalise]{cleveref}

\usepackage{stmaryrd, tikz-cd, enumitem, mathrsfs, bbm, algorithm, url, esint, listings, moreverb, makecell, thm-restate, soul}

\def\bE{\mathbb{E}}

\def\bP{\mathbb{P}}
\def\bQ{\mathbb{Q}}

\def\bZ{\mathbb{Z}}

\def\cA{\mathcal{A}}
\def\cB{\mathcal{B}}
\def\cC{\mathcal{C}}
\def\cD{\mathcal{D}}
\def\cE{\mathcal{E}}

\def\cP{\mathcal{P}}

\def\contig{\vartriangleleft}

\def\wt{\widetilde}

\DeclareMathOperator\Ber{\mathrm{Ber}}
\DeclareMathOperator\Bin{\mathrm{Bin}}
\DeclareMathOperator\BSC{\mathrm{BSC}}

\DeclareMathOperator\dist{\mathrm{dist}}

\DeclareMathOperator\KL{\mathrm{KL}}

\DeclareMathOperator\supp{\mathrm{supp}}

\DeclareMathOperator\Unif{\mathrm{Unif}}

\DeclareMathOperator\Inf{\mathsf{Inf}}
\DeclareMathOperator\I{\mathsf{I}}
\DeclareMathOperator\bs{\mathsf{bs}}
\DeclareMathOperator\s{\mathsf{s}}
\DeclareMathOperator\R{\mathsf{R}}

\renewcommand{\setminus}{\backslash}
\DeclareMathOperator*{\argmax}{arg\,max}

\newcommand{\rom}[1]{\textup{\uppercase\expandafter{\romannumeral#1}}}

\newtheorem{theorem}{Theorem}[section]
\newtheorem{lemma}[theorem]{Lemma}
\newtheorem{proposition}[theorem]{Proposition}
\newtheorem{corollary}[theorem]{Corollary}

\theoremstyle{definition}
\newtheorem{definition}[theorem]{Definition}

\newtheorem{open}{Open Problem}

\newcommand{\DKLs}[0]{D_{\KL}}
\newcommand{\DKL}[0]{\DKLs(p \parallel 1-p)}
\newcommand{\Threshold}[2]{\mathrm{TH}_{#2}^{#1}}

\begin{document}
\title{Tight Bounds for Noisy Computation of High-Influence Functions, Connectivity, and Threshold}
\ifdefined\isarxiv
\author{Yuzhou Gu\thanks{\texttt{yuzhougu@nyu.edu}. New York University. Part of the work was done while supported by the National Science Foundation under Grant No.~DMS-1926686.}
\and
Xin Li\thanks{\texttt{lixints@cs.jhu.edu}. Johns Hopkins University. Supported by NSF CAREER Award CCF-1845349 and NSF Award CCF-2127575.}
\and
Yinzhan Xu\thanks{\texttt{xyzhan@ucsd.edu}. University of California, San Diego. Supported by NSF Grant CCF-2330048, HDR TRIPODS Phase II grant 2217058 and a Simons Investigator Award.}}
\else
\author{Anonymous Author(s)}
\fi
\date{}

\maketitle

\pagenumbering{gobble}

\begin{abstract}
  The noisy query model (a.k.a.~noisy decision tree model) was formally introduced by [Feige, Raghavan, Peleg and  Upfal, SICOMP 1994].
  In this model, the (binary) return value of every query (possibly repeated) is independently flipped with some fixed probability $p \in (0, 1/2)$.
  The noisy query complexity (a.k.a.~noisy decision tree complexity) is the smallest number of queries an algorithm needs to make to achieve error probability $\le \frac 13$.

  Previous works often focus on specific problems, and it is of great interest to have a characterization of noisy query complexity for general functions. We show that any Boolean function with total influence $\Omega(n)$ has noisy query complexity $\Omega(n\log n)$, matching the trivial upper bound of $O(n \log n)$. Our result is the first noisy query complexity lower bound of this generality, beyond what was known for random Boolean functions [Reischuk and Schmeltz, FOCS 1991].

  Our second contribution is to determine the asymptotic noisy query complexity of the Graph Connectivity problem. In this problem, the goal is to determine whether an undirected graph is connected, where each query asks for the existence of an edge in the graph.
  A simple algorithm can solve the problem with error probability $o(1)$ using $O(n^2 \log n)$ noisy queries, but no non-trivial lower bounds were known prior to this work.
  We show that any algorithm that solves the problem with error probability $\le \frac 13$ uses $\Omega(n^2\log n)$ noisy queries.

  For the proofs of the above lower bounds, we develop a three-phase method, which is a refinement of the two-phase method of [Feige, Raghavan,  Peleg and  Upfal, SICOMP 1994]. Our three-phase method adds an extra step where carefully designed ``free'' information is disclosed to the algorithm. This way, we can precisely control the posterior distribution of the input given the query results, empowering a more refined analysis.

  Last but not least, we determine the exact number of noisy queries (up to lower order terms) needed to solve the $k$-Threshold problem and the Counting problem. The $k$-Threshold problem asks to decide whether there are at least $k$ ones among $n$ bits, given noisy query access to the bits.
  We prove that $(1\pm o(1)) \frac{n\log (\min\{k,n-k+1\}/\delta)}{(1-2p)\log \frac{1-p}p}$ queries are both sufficient and necessary to achieve error probability $\delta = o(1)$. Previously, such a result was only known when $\min\{k,n-k+1\}=o(n)$ [Wang, Ghaddar, Zhu and Wang, arXiv 2024]. We also show a similar $(1\pm o(1)) \frac{n\log (\min\{k+1,n-k+1\}/\delta)}{(1-2p)\log \frac{1-p}p}$ bound for the Counting problem, where one needs to count the number of ones among $n$ bits given noisy query access and $k$ denotes the answer.
\end{abstract}

\newpage
\pagenumbering{arabic}

\section{Introduction}

Fault tolerance is a crucial feature of algorithms that work for large systems, as errors occur unavoidably. Hence, previous studies have considered various models to capture the effect of errors, such as R{\'e}nyi--Ulam game \cite{renyi1961problem, ulam1976adventures}, independent noise \cite{feige1994computing}, and independent noise without resampling \cite{braverman2008noisy}.

\cite{feige1994computing} formally proposed the noisy query model with independent noise (which they call the noisy Boolean decision tree model or the noisy comparison tree model, depending on whether the problem uses point queries to input bits or comparison queries between input elements). In this model, each query returns a bit (either an input bit or a pairwise comparison result) that is independently flipped with some fixed probability $p \in (0, \frac{1}{2})$ (i.e. independent noise) and repeated queries are allowed. The efficiency of an algorithm is measured in terms of the number of queries it makes. \cite{feige1994computing}  showed tight asymptotic bounds for the noisy query complexity for a wide range of problems, including Parity, Threshold, Binary Search and Sorting.

In fact, researchers had studied queries with independent noise even before \cite{feige1994computing} formally defined the model. \cite{berlekamp1964block,horstein1963sequential,burnashev1974interval} all studied some versions of Binary Search under independent noise. In particular, \cite{berlekamp1964block,horstein1963sequential} studied the problem through the lens of channel coding (see \cite{wang2022noisy} for a more detailed discussion about the relationship between the channel coding perspective and the noisy query perspective). These examples further demonstrate that the noisy query model by~\cite{feige1994computing} is a natural model to study.

Following \cite{feige1994computing}, researchers have studied problems in the noisy query model extensively, including random functions \cite{reischuk1991reliable,feige1992complexity,evans1998average}, $k$-CNF and $k$-DNF \cite{DBLP:journals/rsa/KenyonK94}, Binary Search \cite{ben2008bayesian,dereniowski2021noisy, gu2023optimal, zhu2023optimal}, Sorting~\cite{wang2022noisy,wang2023variable,gu2023optimal, zhu2023optimal}, Graph Search \cite{Emamjomeh-Zadeh16, DereniowskiTUW19, dereniowski2021noisy} (a generalization of Binary Search), and $k$-Threshold \cite{zhu2023optimal, zhu2023noisy,wang2024noisy}.

However, despite the popularity and naturality of the model, most research on the noisy query model focus on specific functions instead of general functions. In the above examples, the only exceptions are the lower bounds for random functions \cite{reischuk1991reliable,feige1992complexity,evans1998average}, and upper bounds for $k$-CNF and $k$-DNF \cite{DBLP:journals/rsa/KenyonK94}. Furthermore, the specific functions studied in literature are often the ones studied already in \cite{feige1994computing} such as Threshold, Binary Search and Sorting. As a result, the noise query complexity of many important and natural problems are left unexplored. In this paper, we take a first step towards studying the noise query complexity of more general functions and problems.

\subsection{High-influence functions}

Our first result is a lower bound for the noisy query complexity of high-influence functions, a greatly general family of functions. This result is the first result towards understanding the lower bound of general Boolean functions, beyond the lower bound for random functions \cite{reischuk1991reliable,feige1992complexity,evans1998average}.

Influence is a central quantity in the analysis of Boolean functions. For a Boolean function $f: \{0,1\}^n\to \{0,1\}$, the influence of coordinate $i\in [n]$ is defined as
\begin{align*}
  \Inf_i(f) =  \bP_{x\sim \{0, 1\}^n}(f(x) \ne f(x\oplus e_i)),
\end{align*}
where $e_i$ denotes the bit string where the only $1$ is on the $i$-th coordinate, and $\oplus$ denotes exclusive or. That is, $\Inf_i(f)$ is the probability that flipping the $i$-th coordinate of a uniformly random bit string also flips the function value.
The \emph{total influence} is the sum of influences over all coordinates, i.e.,
\begin{align*}
  \I(f) = \sum_{i\in [n]} \Inf_i(f).
\end{align*}
We prove that Boolean functions with linear total influence have noisy query complexity $\Omega(n\log n)$.

\begin{restatable}[Noisy query complexity of high-influence functions]{theorem}{LinearTotalInfluence} \label{thm:influence}
  For any $c>0$, there exists $c'>0$ such that the following holds.
  For any Boolean function $f: \{0,1\}^n\to \{0,1\}$ with $\I(f) \ge c n$, any noisy query algorithm computing $f(x)$ with error probability $\le \frac 13$ makes at least $c' n\log n$ noisy queries in expectation to the coordinates of the input $x \in \{0, 1\}^n$.
\end{restatable}
Note that the error probability $\frac 13$ can be replaced with any $0<\epsilon<\frac 12$ without affecting the asymptotic noisy query complexity. The statement is tight in the sense that any Boolean function on $n$ inputs can be computed with error probability $o(1)$ using $O(n\log n)$ noisy queries: by simply querying each bit $O(\log n)$ times, we can determine the input string with $o(1)$ error probability.

\cref{thm:influence} unifies and generalizes several previous results. For example, it was known that a random Boolean function (with probability $1-o(1)$) has noisy query complexity $\Omega(n\log n)$ \cite{reischuk1991reliable,feige1992complexity,evans1998average}, and computing the parity function requires $\Omega(n\log n)$ noisy queries \cite{feige1994computing}. As random Boolean functions and the parity function have total influence $\Omega(n)$, \cref{thm:influence} immediately implies these lower bounds as special cases.

Another central notion related to influence is the sensitivity of Boolean functions\footnote{The sensitivity of a Boolean function $f$ at input $x$, denoted by $\s(f, x)$, is the number of bits $i$ for which $f(x) \ne f(x\oplus e_i)$. }, as it is well known that the total influence is the same as the \emph{average sensitivity} (the expected sensitivity over a uniformly random input). \cite{reischuk1991reliable} proved that any non-adaptive algorithm computing a Boolean function $f$ makes at least $\Omega(\s(f)\log \s(f))$ noisy queries, where $\s(f)$ is the (maximum) sensitivity of $f$. This result is incomparable to \cref{thm:influence}, as their lower bound only holds against non-adaptive algorithms. In fact, it is not possible to extend the $\Omega(\s(f)\log \s(f))$ lower bound against adaptive algorithms in general. For instance, the $\mathrm{OR}$ function has sensitivity $n$ and (adaptive) noisy query complexity $O(n)$ \cite{feige1994computing}. On the other hand, the average sensitivity of $\mathrm{OR}$ is much smaller, which suggests that the average sensitivity of a Boolean function $f$ is a more reasonable measure to lower bound the \emph{adaptive} noisy query complexity. This motivates us to raise the following open question, towards a lower bound for general Boolean functions.

\begin{open} \label{open:influence}
Is it true that every Boolean function $f: \{0, 1\}^n \rightarrow \{0, 1\}$ has noisy query complexity $\Omega(\I(f) \log \I(f))$?
\end{open}
\cref{thm:influence} resolves the case where $\I(f) = \Omega(n)$. We note another evidence supporting the $\I(f) \log \I(f)$ lower bound. The randomized query complexity $\R(f)$ satisfies $\R(f) = \Omega(\bs(f)) = \Omega(\s(f)) = \Omega(\I(f))$, where $\bs(s)$ denotes block sensitivity and the first step is by \cite[Lemma 4.2]{nisan1989crew}. In general, the noisy query complexity of a function $f$ is always between $\R(f)$ and $O(\R(f) \log \R(f))$. Therefore, the $\Omega(\I(f) \log \I(f))$ lower bound is consistent with these known relationships.

For the proof of \cref{thm:influence}, we develop a three-phase lower bound framework, which is based on and refines \cite{feige1994computing}'s two-phase method for proving a lower bound for the $k$-Threshold problem. In the three-phase method, we reduce the original problem in the noisy query model to a stronger observation model, where in Phase 1 the algorithm makes non-adaptive noisy observations and in Phase 3 the algorithm makes adaptive exact observations.
In Phase 2, the model gives away free information, which can only help the algorithm.
By designing this free information carefully, the effect of Phase 1 and 2 combined can be significantly simplified, allowing for a precise analysis in Phase 3.

We note that this idea of giving away free information already appears in \cite{feige1994computing}'s two-phase method.
For their problem ($k$-Threshold), this free information is relatively simple.
However, for other and more general problems, the free information could be significantly more involved.
We design the free information in a different way in order for the analysis in Phase 3 to be viable.
This additional phase to the original two-phase method makes it easier to apply and allows for other applications.
As we will soon discuss, the three-phase framework is essential for our result on Graph Connectivity and also leads to a simple proof for the lower bound of $k$-Threshold.

\subsection{Graph Connectivity}
Although the noisy query model is quite natural, there has been little prior work studying graph problems in this model. Some prior examples include \cite{feige1994computing}, which briefly mentioned that a lower bound for the noisy query complexity of Bipartite Matching can be achieved by reducing from the other problems they studied. \cite{DBLP:journals/rsa/KenyonK94} designed algorithms for $k$-CNF and $k$-DNF using a small number of queries, which imply, for instance, that one can test, up to error probability $\delta$, whether a given $n$-vertex graph contains a triangle using $O\left(n^2 \log \frac{1}{\delta}\right)$ noisy queries.

One of the most fundamental graph problems is Graph Connectivity, where we are given an $n$-vertex undirected graph $G$, and need to determine whether the graph is connected via edge queries. It is so basic that those studying algorithms encounter it very early on. For instance, breadth-first-search and depth-first-search are usually among the first graph algorithms taught in undergraduate algorithm classes, and the simplest application of them is to detect whether a graph is connected. However, to our surprise, we do not even have a good understanding of the noisy query complexity of such an elementary problem.

One simple algorithm for Graph Connectivity is to query every edge in the input graph $O(\log n)$ times to correctly compute the input graph with high probability, and then solve Graph Connectivity on the computed graph. This naive algorithm uses $O(n^2 \log n)$ noisy queries, and is essentially all what was previously known about Graph Connectivity in the noisy query model.
In particular, hardness of Graph Connectivity does not seem to follow from known hardness results. %

Using the three-phase method, we prove an $\Omega(n^2 \log n)$ lower bound on the noisy query complexity of Graph Connectivity, showing that the naive $O(n^2 \log n)$ algorithm is actually optimal up to a constant factor:

\begin{restatable}[Hardness of Graph Connectivity]{theorem}{GraphConnectivity}\label{thm:graph-conn-hard}
  Any algorithm solving the Graph Connectivity problem with error probability $\le \frac 13$ uses $\Omega(n^2\log n)$ noisy queries in expectation.
\end{restatable}
Similarly as before, the error probability $\frac 13$ can be replaced with any $0<\epsilon<\frac 12$ without affecting the asymptotic noisy query complexity.

We also show an $\Omega(n^2 \log n)$ lower bound for the related $s$-$t$ Connectivity problem, where we are given an $n$-vertex undirected graph $G$ and two fixed vertices $s$ and $t$, and the goal is to determine whether there is a path in the graph connecting $s$ and $t$.

As Graph Connectivity and $s$-$t$ Connectivity are very basic tasks on graphs, their lower bounds immediately imply lower bounds for several other fundamental graph problems as well. For instance, given the lower bounds for Graph Connectivity and $s$-$t$ Connectivity, it is straightforward to show that Global Min-Cut, $s$-$t$ Shortest Path, and $s$-$t$ Max Flow on unweighted undirected graphs all require $\Omega(n^2 \log n)$ noisy queries in expectation.

\subsection{Threshold and Counting}

In the $k$-Threshold problem, one is given a length-$n$ Boolean array $a$ and an integer $k$, and the goal is to determine whether the number of $1$'s in the array $a$ is at least $k$. Note that the answer to the input is false if and only if the number of $0$'s in the input is at least $n - k + 1$. We can thus solve $k$-Threshold using an algorithm for $(n-k+1)$-Threshold: we can flip all input bits, change $k$ to $n - k + 1$, solve the modified instance, and finally flip the result. Therefore, we can assume without loss of generality that $k \le n - k +  1$, or equivalently, $k \le (n + 1) / 2$.

$k$-Threshold is one of the first problems studied in the noisy query model. In \cite{feige1994computing}, it was shown that $\Theta\left(n \log \frac k\delta\right)$ queries are both sufficient and necessary to solve the problem with error probability $\delta$. However, the optimal constant factor was left unknown.\footnote{Studying constant factors is often overlooked in theoretical computer science, but in this research area, determining the optimal constants for noisy query complexities of other fundamental problems such as Binary Search and Sorting has been an active topic (e.g., \cite{burnashev1974interval, ben2008bayesian,dereniowski2021noisy, gu2023optimal}). See \cite{DBLP:conf/icalp/Gretta024} for more discussions on the importance of studying constants in query complexity. }

There has been some progress towards determining the exact constant for $k$-Threshold. In~\cite{zhu2023noisy}, it was shown that the noisy query complexity of the $\mathrm{OR}$ function on $n$ input bits (equivalent to $1$-Threshold) is
\[
(1 \pm o(1)) \frac{n \log \frac{1}{\delta}}{\DKL}
\]
for $\delta=o(1)$,  where $\DKL=(1-2p)\log \frac{1-p}p$ is the Kullback-Leibler divergence between two Bernoulli distributions with heads probabilities $p$ and $1 - p$.
This result was later generalized to $k$-Threshold for all $k = o(n)$ and $\delta = o(1)$ by \cite{wang2024noisy}, who showed an
\[
(1 \pm o(1)) \frac{n \log \frac{k}{\delta}}{\DKL}
\]
bound. Compared to \cite{zhu2023noisy}, \cite{wang2024noisy}'s result works for a much wider range of $k$. However, their lower bound proof technique unfortunately stops working for the case $k = \Theta(n)$, and this case is frustratingly left open (we remark that their algorithm gives the right upper bound even for $k=\Theta(n)$).

In this work, we complete the last piece of the puzzle, showing a matching bound for all values of $k$.
\begin{restatable}[Noisy query complexity of $k$-Threshold]{theorem}{KThreshold}
    \label{thm:k-threshold}
For any $1 \le k \le (n + 1) / 2$ and $\delta = o(1)$, computing $k$-Threshold on a length-$n$ array with error probability $\delta$ needs and only needs
\[
(1 \pm o(1)) \frac{n \log \frac{k}{\delta}}{\DKL}
\]
noisy queries in expectation.
\end{restatable}
Here the $\delta=o(1)$ assumption is standard and has appeared in several previous works \cite{gu2023optimal,zhu2023noisy,wang2024noisy}.

While \cite{wang2024noisy} has already given an algorithm achieving the tight upper bound for any $k$, their algorithm involves calling some extra algorithms such as Noisy Sorting and Noisy Heap, which seems too heavy and unnecessary for the $k$-Threshold problem (after all, in the classic noiseless setting, the algorithm for $k$-Threshold is much simpler than algorithms for Sorting or Heap). We provide a much simpler algorithm which involves only checking each bit one by one and completely avoids calling these extra algorithms.

We also provide an alternative and simpler proof of the lower bound for $k$-Threshold with $k=o(n)$. The proof of \cite{wang2024noisy} considers three cases and uses two different methods (the two-phase method from \cite{feige1994computing} and Le Cam's two point method) for solving them. We show that this casework is unnecessary by providing a uniform and simple proof for all $k=o(n)$ by using our three-phase method.

We also consider a related problem, Counting, where we need to compute the number of $1$'s in $n$ input Boolean bits. The lower bound for $k$-Threshold easily applies to Counting as well (though in a non-black-box way). In addition, we design an algorithm for Counting that matches the lower bound, obtaining the following result.
\begin{restatable}[Noisy query complexity of Counting]{theorem}{Counting}
    \label{thm:counting}
    Given a sequence $a \in \{0, 1\}^n$, computing $\lVert a \rVert_1$ with error probability $\delta = o(1)$ needs and only needs
    \[
    (1\pm o(1))\frac{n \log \frac{\min\{\lVert a \rVert_1, n - \lVert a \rVert_1\}+1}{\delta}}{\DKL}
    \]
    noisy queries in expectation.
\end{restatable}

A problem closely-related to $k$-Threshold is the $k$-Selection problem, where one is given $n$ items (comparable with each other) and the goal is to select the $k$-th largest element using noisy comparison queries. It is known that solving $k$-Selection with error probability $\delta=o(1)$ needs and only needs $\Theta\left(n\log \frac {\min\{k,n-k+1\}}\delta\right)$ noisy queries \cite{feige1994computing}.
Their bounds are only tight up to a constant factor, so the exact value of the leading coefficient remains open.
\begin{open} \label{open:k-selection}
  Determine the exact constant $c$ such that $(c\pm o(1)) n\log \frac {\min\{k,n-k+1\}}{\delta}$ noisy queries is both sufficient and necessary to solve the $k$-Selection problem with error probability $\delta=o(1)$.
\end{open}

\section{Technical overview}
\subsection{Proof overview for high-influence functions}
\label{sec:influence:overview}
Recall that the error probability $\frac 13$ in the statement of \cref{thm:influence} can be replaced with any $0<\epsilon<\frac 12$ without loss of generality. Also, the expected number of queries can be replaced by the worst-case number of queries by Markov's inequality.

Let $f: \{0, 1\}^n \rightarrow \{0, 1\}$ be a Boolean function with $\I(f) = \Omega(n)$. The hard distribution for the input $x$ will be the uniform distribution over $\{0, 1\}^n$.

Inspired by \cite{feige1994computing}, we prove hardness under the noisy query model by introducing a new problem where the algorithm has more power, and prove hardness of this new problem.

The new problem has three phases, described as follows.
\begin{enumerate}
  \item \label{item:sec:influence:overview:phase-1} In Phase 1, the algorithm makes $m_1=c_1\log n$ noisy queries to coordinate $x_i$ for every $i \in [n]$.
  \item \label{item:sec:influence:overview:phase-2} In Phase 2, the oracle reveals some coordinates of $x$ to the algorithm.
  \item \label{item:sec:influence:overview:phase-3} In Phase 3, the algorithm makes $m_2=c_2 n$ adaptive exact queries for some constant $c_2$.
\end{enumerate}
The goal of the algorithm is to compute the value $f(x)$.

Note that the first two phases are non-adaptive. The third phase is adaptive but the algorithm makes exact queries. This is the reason why the three-phase problem is easier to analyze than the original noisy query problem.

It is not hard to prove that if no algorithm can solve the three-phase problem with $\epsilon$ error probability, then no algorithm can solve the original problem with $\epsilon$ error probability using no more than $m_1 m_2 = c_1 c_2 n \log n$ noisy queries.
Therefore we only need to prove hardness of the three-phase problem.

\paragraph{Phase 1.} Let $a_i$ be the number of times where a query to $x_i$ returns $1$ in Phase $1$. The posterior distribution of the input $x$ given observations made in Phase $1$ depends only on the variables $a_i$. In other words, $(a_i)_{i \in [n]}$ is a sufficient statistic for $x$. Conditioned on $x$, the variables $(a_i)_{i \in [n]}$ are independent, and the distribution of each $a_i$ is a binomial distribution depending only on whether $x_i = 1$. That is, if $x_i = 1$, then $a_i \sim \Bin(m_1, 1 - p)$; otherwise, $a_i \sim \Bin(m_1, p)$ ($\Bin(\cdot, \cdot)$ denotes the binomial distribution). Moreover, for $x_i = 0$, $a_i$ is in an interval $I$ around $p m_1$ with probability $1-o(1)$; for $x_i = 1$, $a_i$ is in the interval $I$ with probability $n^{-c_3\pm o(1)}$, for some constant $c_3>0$ depending on $c_1$ and $p$.
Because the observations are independent for different coordinates, each index has a non-negative weight, such that the posterior probability of the input being $x$ is proportional to the product of weights of the coordinates $i$ where $x_i = 1$.

\paragraph{Phase 2.} In Phase 2, the oracle reveals some coordinates of $x$ to the algorithm. This information is revealed in two steps:
\begin{enumerate}[label=2\alph*.]
\item All coordinates with $a_i \not \in I$ are revealed.
\item Every $x_i$ with $a_i \in I$ is revealed independently with probability $q_{a_i}$ for some real numbers $(q_k \in [0, 1])_{k \in I}$.
\end{enumerate}
Because of the observations in Phase 1, the unrevealed coordinates can have different weights.
That is, given observations up to Step 2a, the posterior probabilities for different coordinates being $1$ can be different, which is undesirable.
Step 2b is a subsampling procedure, with the goal of reweighting the unrevealed coordinates so that all of them have the same weight.
If the interval $I$ is not too large, then the probabilities $q_k$ for $k\in I$ will not be too small.
Because observations made up to Step 2b are independent for different coordinates, they have the same effect as the following procedure (if the real numbers $q_k$ are chosen carefully): every $x_i$ with $x_i = 1$ is revealed independently with probability $p_+=1-n^{-c_3\pm o(1)}$ and $x_i$ with $x_i = 0$ is revealed independently with probability $p_-=o(1)$.

\paragraph{Phase 3.} At this stage, let the set of unrevealed coordinates be $U$. Conditioned on the revealed coordinates, $x_i$ for $i\in U$ are i.i.d.~$\Ber(q)$ variables, where  $\Ber(q)$ denotes the Bernoulli distribution with head probability $q=\frac{1-p_+}{1-p_++1-p_-}=n^{-c_3\pm o(1)}$. In other words, the distribution of $x_U$ is $\Ber(q)^{\otimes U}$. Let $g$ be a restriction of $f$ where the revealed coordinates of $x$ are fixed to be the revealed values. Then we need to show that on average, computing $g$ with error probability $\le \epsilon$ requires $\Omega(n)$ (adaptive) exact queries, for some sufficiently small $\epsilon > 0$. To this end, we consider a biased version of total influence $\I_q$:
\begin{align*}
  \I_q(g) = \sum_{i\in U} \bP_{x\sim \Ber(q)^{\otimes U}}( g(x) \ne g(x\oplus e_i)).
\end{align*}
Our proof strategy consists of the following three steps.
\begin{enumerate}[label=3\alph*.]
  \item First, we show that $\I_q(g) = \Omega(n)$ in expectation.
  \item After we make a query, we further fix the value of the queried coordinate, and replace $g$ with the new restricted function. We show that each exact query can only decrease $\I_q(g)$ by at most $1$ in expectation.
  \item Finally, we show that if $\I_q(g) = \Omega(n)$, then it is impossible to guess $g(x)$ (where $x$ follows a biased product distribution) with error probability $\le \epsilon$.
\end{enumerate}
Combining the three steps, we obtain that making $c_2n$ exact queries (in expectation) cannot compute $g$ with $\le \epsilon$ error probability, for sufficiently small $c_2$ and $\epsilon > 0$.

We note that Step 3c is where this approach fails to extend to general functions with sublinear total influence. In fact, for $\I_q(g) = o(n)$, it might be possible to guess $g(x)$ (where $x$ follows a biased product distribution) with $o(1)$ error probability. For instance, if $g(x)$ is a function where $o(1)$ random fraction of the values are $0$, and the other values are $1$, it is straightforward to show that $\I_q(g) = o(n)$ (in expectation), and we can guess $g(x)=1$ to achieve $o(1)$ error probability, without using any queries.

\subsection{Proof overview for Graph Connectivity} \label{sec:conn:overview}
In this section we present an overview of our proof of \cref{thm:graph-conn-hard}, hardness of Graph Connectivity. Again, the error probability $\frac 13$ can be replaced with any $0<\epsilon<\frac 12$ and the expected number of queries can be replaced by the worst-case number of queries, without loss of generality.

\paragraph{Hard distribution.}
To prove the lower bound, we design a hard distribution of inputs.
The distribution is based on uniform random spanning trees (USTs) of the complete graph.
Let $T$ be a UST. We say an edge $e\in T$ is $\beta$-balanced if both components of $T\backslash e$ have size at least $\beta n$.
Let $e$ be a uniform random $\beta_0$-balanced edge of $T$, where $\beta_0=\frac 1{21}$. (If such $e$ does not exist, we restart.)
We throw a fair coin and decide whether to erase edge $e$ from $T$.
That is, conditioned on $(T,e)$, the input graph $G$ is $T$ (connected) with probability $\frac 12$, and is $T\backslash e$ (disconnected) with probability $\frac 12$.

\paragraph{Three-phase problem.}
Following the three-phrase method, it suffices to show the hardness of the three-phase problem described as follows.
\begin{enumerate}[label=\arabic*.]
  \item \label{item:sec:conn:overview:phase-1} In Phase 1, the algorithm makes $m_1=c_1\log n$ noisy queries to every unordered pair (called ``potential edge'') $(u,v)\in \binom V2$ for some constant $c_1$.
  \item \label{item:sec:conn:overview:phase-2} In Phase 2, the oracle reveals some edges and non-edges of $G$ to the algorithm.
  \item \label{item:sec:conn:overview:phase-3} In Phase 3, the algorithm makes $m_2=c_2 n^2$ adaptive exact queries for some constant $c_2$.
\end{enumerate}
The goal of the algorithm is to determine whether the graph is connected.

\paragraph{Phase 1.}
This phase is similar to Phase 1 in \cref{sec:influence:overview}.
In Phase 1, the algorithm makes $m_1$ noisy queries to every potential edge $e\in \binom V2$.
Let $a_e$ be the number of times where a query to a potential edge $e$ returns $1$ in Phase 1.
Similar to \cref{sec:influence:overview}, if $e\in G$, then $a_e \sim \Bin(m_1, 1-p)$; otherwise $a_e \sim \Bin(m_1, p)$.
Specifically, for $e\not \in G$, $a_e$ is in an interval $I$ around $p m_1$ with probability $1-o(1)$; for $e\in G$, $a_e$ is in the interval $I$ with probability $n^{-c_3\pm o(1)}$, for some constant $c_3>0$ depending on $c_1$ and $p$.
Because the observations are independent for different edges, each edge has a non-negative weight, such that the posterior probability of $G=T$ for a tree $T$ is proportional to the product of weights of edges in $T$.

\paragraph{Phase 2.}
In Phase 2, the oracle reveals some edges and non-edges of $G$.
This information is revealed in three steps (the first two steps are similar to those in \cref{sec:influence:overview}).
\begin{enumerate}[label=2\alph*.]
  \item \label{item:sec:conn:overview:step-2a} In Step 2a, the potential edges $e$ with $a_e \not \in I$ are revealed. (Recall that $I$ is an interval around $p m_1$.) That is, the algorithm now knows which potential edges $e$ with $a_e \not \in I$ are edges.
  \item \label{item:sec:conn:overview:step-2b} In Step 2b, every edge $e$ with $a_e\in I$ is revealed independently with probability $q_{a_e}$, for some real numbers $(q_k\in [0,1])_{k\in I}$.
  \item \label{item:sec:conn:overview:step-2c} In Step 2c, $n-2$ edges are revealed as follows.
  If $G$ is disconnected, reveal all edges of $G$; otherwise, if there is a $\beta_0$-balanced edge that has not been revealed so far, uniformly randomly choose one (say $e^*$) from all such edges, and reveal all edges of $G\backslash e^*$. If $G$ is connected but all $\beta_0$-balanced edges have been revealed, report failure.
\end{enumerate}

Similar to \cref{sec:influence:overview}, if the real numbers $q_k$ are chosen carefully, the observations made up to Step 2b have the same effect as the following procedure: observe every edge independently with probability $p_+=1-n^{-c_3\pm o(1)}$; observe every non-edge independently with probability $p_-=o(1)$.

We can show that Step 2c reports failure with $1-\Omega(1)$ probability.
In the following, we condition on the event that the oracle does not report failure in Step 2c.
In this case, Step 2c reveals all except for one edge, which forms two connected components $T_1$ and $T_2$.
By the construction, $T_1$ and $T_2$ both have size at least $\beta_0 n$.
The posterior distribution is supported on $\{T_1\cup T_2, T_1\cup T_2\cup \{e\}: e\in E(T_1,T_2)\}$.
In our full analysis, we compute an exact formula for the posterior probability of every graph in the support.

\paragraph{Phase 3.}
In Phase 3 the algorithm makes $m_2=c_2 n^2$ adaptive exact queries.
Let $G_0 = T_1\cup T_2$, $G_e = T_1\cup T_2 \cup \{e\}$ for $e\in E(T_1,T_2)$.
Let $\bP^{(2)}$ denote the posterior distribution of $G$ after Phase 2.
Because $T_1$ and $T_2$ are already revealed, we can w.l.o.g.~assume that the algorithm makes queries only to edges in $E(T_1,T_2)$.
We prove that after Phase 2, for most edges $e\in E(T_1,T_2)$, we have
\begin{align*}
  \frac{\bP^{(2)}(G_e)}{\bP^{(2)}(G_0)} = \Theta\left(\frac 1{|T_1||T_2|}\right) = \Theta\left(\frac 1{n^2}\right)
\end{align*}
Let $E^{(3)}$ be the set of potential edges queried in Phase 3.
For small enough  $c_2$, we have
\begin{align*}
  \sum_{e\in E^{(3)}} \bP^{(2)}(G_e) = O(1) \cdot \bP^{(2)}(G_0).
\end{align*}
Therefore, with constant probability, all queries in Phase 3 return $0$.
Furthermore, if this happens, then the posterior probability of $G_0$ and $\{G_e: e\in E(T_1,T_2)\backslash E^{(3)}\}$ are within a constant factor of each other.
In this situation, outputting anything will result in a constant error probability.
This concludes that for some $\epsilon>0$, no algorithm can solve the three-phase problem with error probability $\epsilon$.

\subsection{Proof overview for \texorpdfstring{$k$}{k}-Threshold and Counting} \label{sec:threshold:overview}

Before we discuss our techniques for $k$-Threshold and Counting, we briefly discuss the previous work of \cite{wang2024noisy}, who showed a
\[
(1 \pm o(1)) \frac{n \log \frac{k}{\delta}}{\DKL}
\]
bound for $k$-Threshold where $k = o(n)$. Their lower bound stops working for $k = \Theta(n)$ because one step in their lower bound proof reveals the locations of $k-1$ $1$'s to the algorithm, leaving only $n-k+1=(1-\Omega(1))n$ unknown bits. This case can be solved by an algorithm using $(1-\Omega(1))\frac{n \log \frac{k}{\delta}}{\DKL}$ noisy queries, meaning that this approach cannot be used to show a tight lower bound.

\paragraph{Lower bound for $k$-Threshold.}
Our lower bound for $k$-Threshold for general values of $k$ is a reduction from~\cite{wang2024noisy}'s lower bound for $k = o(n)$. In the overview, we focus on the case where $k = (n + 1) / 2$ for odd $n$. In this case, the problem is equivalent to computing the majority of $n$ input bits.

Given any instance of $k$-Threshold on a length-$n$ array for $k = n / \log n$, we first add $L$ artificial $1$'s to the array to obtain a new instance where $n' = n + L$ and $k' = k + L$. We set $L$ so that $k' = (n' + 1) / 2$ (or equivalently, $L = n - 2k + 1$). Now suppose we have an algorithm for the new instance that uses only $(1-\epsilon) \frac{n' \log \frac{k'}{\delta}}{\DKL}$ noisy queries for some $\epsilon > 0$. Whenever the algorithm queries an artificial $1$, it can be simulated without incurring an actual noisy query; instead, we only need to flip biased coin with head probability $1-p$ and return its value. Because the algorithm is for computing the majority, intuitively, by symmetry, the expected number of queries it spends on an input $0$ and an input $1$ should be the same. Furthermore, if we add the artificial $1$'s to random positions, the algorithm should not be able to distinguish an artificial $1$ with an actual $1$. Therefore, in expectation, $L / n'$ fraction of the algorithm's queries are to artificial $1$'s, so the actual query complexity for solving the original $k$-Threshold instance is
\[
\left(1-\frac{L}{n'}\right) \cdot (1-\epsilon) \frac{n' \log \frac{k'}{\delta}}{\DKL} = (1-\epsilon)\frac{n \log \frac{k'}{\delta}}{\DKL}.
\]
Because $k = n / \log n$, we have $\log k' \le \log(n + 1) = (1+o(1)) \log k$, so the above bound becomes
\[
 (1-\epsilon + o(1))\frac{n \log \frac{k}{\delta}}{\DKL},
\]
which contradicts the lower bound from \cite{wang2024noisy}.

Our lower bound for more general values of $k$ is proved using a similar idea. However, we no longer have the symmetry between $0$ and $1$, so we need to reduce from the case $k = n / \log n$ or $k = (n + 1) / 2$ depending on whether the algorithm spends more queries on an input $0$ or $1$.

\paragraph{Upper bound for $k$-Threshold.}
For the upper bound, \cite{wang2024noisy}'s algorithm already works also for the $k = \Theta(n)$ case. Nevertheless, we provide a much simpler algorithm that achieves the same tight upper bound. \cite{wang2024noisy} used a standard $\textsc{Check-Bit}$ procedure to estimate the value of each input bit, and then used established machinery on Noisy Sorting and Noisy Heap studied earlier \cite{feige1994computing}. In comparison, our algorithm uses an asymmetric version of the $\textsc{Check-Bit}$ procedure that estimates the value of each input bit. Using this asymmetric procedure, we essentially only need to check each input bit one by one, avoiding calling extra algorithms such as Noisy Sorting or Noisy Heap.

\paragraph{Upper bound for Counting.}
Our algorithm for Counting is based on the idea of our algorithm for $k$-Threshold.
Our algorithm for $k$-Threshold can additionally count the number of $1$'s in the input if it is at most $k$, so one natural idea for Counting is to first compute an estimation $k$ of the answer $\lVert a \rVert_1$, then use our algorithm for $k$-Threshold to compute the exact answer. However, this approach does not work when $\lVert a \rVert_1$ is very small compared to $n$, as there is no reliable way to estimate the answer within a constant factor using $o\left(n\frac{\log \frac{\lVert a\rVert_1}{\delta}}{\DKL}\right)$ queries.
We circumvent this issue by gradually increasing $k$ during the algorithm and simulate the asymmetric $\textsc{Check-Bit}$ procedure on each input bit.
We can view the asymmetric $\textsc{Check-Bit}$ procedure for each bit as a biased random walk on $\bZ$, and for different $k$ the procedure only changes the stopping condition, but not the random walk. In this way we show that $k$ will eventually stop at the correct answer with desired error probability.

\section{Preliminaries}
Throughout the paper, we use $p \in (0, 1/2)$ to denote the flipping probability of each noisy query, i.e., for a bit $x$, $\textsc{Query}(x)$ returns $x$ with probability $1-p$ and returns $1 - x$ with probability $p$. For $0 \le q \le 1$, $\Ber(q)$ denotes the Bernoulli distribution with head probability $q$.
Throughout the paper, all $\log$s have base $e$.
For two sequences $(f_n)_n$, $(g_n)_n$, we write $f_n \asymp g_n$ if $f_n = \Theta(g_n)$, i.e., there exists $\epsilon>0$ such that $\epsilon f_n \le g_n \le \epsilon^{-1} f_n$ for all $n$ large enough.

Let $\Bin(n,p)$ denote the binomial distribution. The following large deviation bound is useful.
\begin{lemma} \label{lem:binomial-large-deviation}
  Let $0<p<\frac 12$ and $0<q<1$.
  Then for large enough $m$ and integer $k = (q\pm o(1)) m$, we have
  \begin{align*}
    \bP(\Bin(m, p) = k) = \exp\left(-(\DKLs(q \parallel p) \pm o(1)) m\right),
  \end{align*}
  where
  \begin{align*}
    \DKLs(a \parallel b) = a \log \frac ab + (1-a) \log \frac{1-a}{1-b}
  \end{align*}
  is the binary KL divergence function.
\end{lemma}
\begin{proof}
  \begin{align*}
    \bP(\Bin(m, p) = k) =&~ \binom mk p^k (1-p)^{m-k} \\
    \nonumber =&~ \exp(m \log m - k \log k - (m-k) \log (m-k) \pm o(m)) \\
    \nonumber &~\cdot \exp(k \log p + (m-k) \log (1-p))\\
    \nonumber =&~ \exp\left( -\left(q \log \frac{q}{p} + (1-q) \log \frac{1-q}{1-p} \pm o(1)\right) m\right) \\
    \nonumber =&~ \exp(-(\DKLs(q \parallel p) \pm o(1)) m).
  \end{align*}
\end{proof}

\section{High-influence functions}
We prove \cref{thm:influence} in this section. Let us recall the theorem statement.

\LinearTotalInfluence*

\subsection{Three-phase problem}
Let $f$ be a Boolean function with total influence $\I(f) \ge cn$ and $x$ be uniformly chosen from $\{0,1\}^n$.
Our goal is to show that any algorithm computing $f(x)$ with error probability $\epsilon$ makes at least $c'n\log n$ queries in expectation.
(The error probability $1/3$ in  \cref{thm:influence} can be replaced with any $\epsilon > 0$ without affecting the asymptotic noisy query complexity, by repeating the algorithm constantly many times and taking the majority vote.)

Let $c_1,c_2>0$ be two absolute constants. We define a three-phase problem as follows.
\begin{enumerate}
  \item In Phase 1, the algorithm makes $m_1 = c_1 \log n$ queries to every coordinate.
  \item In Phase 2, the oracle reveals some elements to the algorithm.
  \item In Phase 3, the algorithm makes $m_2 = c_2 n$ adaptive exact queries.
\end{enumerate}
The goal of the algorithm is to determine the value of $f(x)$.

\begin{lemma} \label{lem:inf:reduction}
  If no algorithm can solve the three-phase problem with error probability $\epsilon>0$, then no algorithm can compute $f(x)$ with error probability $\epsilon$ using at most $c_1 c_2 n\log n$ noisy queries.
\end{lemma}
\begin{proof}
  Suppose there is an algorithm $\cA$ that computes $f(x)$ with error probability $\epsilon$ using at most $m_1 m_2$ noisy queries.
  We define an algorithm $\cA'$ that solves the three-phase problem with the same error probability.

  Let $x\in \{0,1\}^n$ be chosen uniformly randomly.
  Algorithm $\cA'$ receives query results in Phase 1 and 2 and enters Phase 3.
  It simulates algorithm $\cA$ on the same input $x$ as follows.
  \begin{enumerate}[label=(\arabic*)]
    \item Initially, $x_i\gets *$ for all $i\in [n]$.
    \item When $\cA$ queries coordinate $i$:
    \begin{enumerate}[label=(\alph*)]
      \item If this is the $k$-th time coordinate $i$ is queried for $k \le m_1$, return the $k$-th noisy query result on coordinate $i$ in Phase 1 to $\cA$.
      \item Suppose coordinate $i$ has been queried more than $m_1$ times. If $x_i=*$, make an exact query to coordinate $i$ in Phase 3 and let $x_i\gets$ the query result.
      Return $\BSC_p(x_i)$ to $\cA$.
    \end{enumerate}
    \item When $\cA$ returns, $\cA'$ returns the same result.
  \end{enumerate}

  Because $\cA$ has error probability $\epsilon$, $\cA'$ also has error probability $\epsilon$.
  It suffices to prove that in Phase 3, $\cA'$ makes at most $m_2$ exact queries.
  Note that the number of exact queries $\cA'$ makes in Phase 3 is equal to the number of coordinates $i$ to which $\cA$ makes more than $m_1$ queries.
  Because $\cA$ makes $m_1 m_2$ queries, the expected number of such coordinates is at most $m_2$.
\end{proof}

By \cref{lem:inf:reduction}, to prove \cref{thm:influence}, it suffices to prove hardness of the three-phase problem.

\begin{proposition} \label{prop:inf:three-phase-hard}
  For some $c_1,c_2,\epsilon>0$, no algorithm can solve the three-phase problem with error probability $\epsilon$, where the input is uniformly chosen from $\{0,1\}^n$.
\end{proposition}
\cref{thm:influence} follows by combining \cref{lem:inf:reduction} and \cref{prop:inf:three-phase-hard}. The rest of the section is devoted to the proof of \cref{prop:inf:three-phase-hard}.

\subsection{Phase 1} \label{sec:inf:phase-1}
In Phase 1, the algorithm makes $m_1 = c_1 \log n$ queries to every element $i\in [n]$.
Let $A = \{i\in [n]: x_i=1\}$, where $x$ is the input bit string.
Let $a_i$ denote the number of times where a query to $i$ returns $1$. Then for $i\in A$, $a_i\sim \Bin(m_1,1-p)$; for $i\not \in A$, $a_i \sim \Bin(m_1,p)$. For $0\le j\le m_1$, define
\begin{align}
  p_j = \bP(\Bin(m_1,1-p)=j) = \binom{m_1}j (1-p)^j p^{m_1-j}.
\end{align}
Let $I = \left[ p m_1 - m_1^{0.6}, p m_1 + m_1^{0.6} \right]$.
\begin{lemma} \label{lem:inf:binom-concentrate}
  Let $x\sim \Bin(m_1, 1-p)$, $y\sim \Bin(m_1, p)$.
  Then
  \begin{align}
    \label{eqn:lem-inf-binom-concentrate:i} \bP(x\in I) &= n^{-c_3\pm o(1)}, \\
    \label{eqn:lem-inf-binom-concentrate:ii} \bP(y\in I) &= 1-o(1),
  \end{align}
  where $c_3 = c_1(1-2p)\log\frac{1-p}p$.
\end{lemma}
\begin{proof}
  By \cref{lem:binomial-large-deviation}, for $k\in I$, we have
  \begin{align*}
    \bP(x=k) = \exp(-(\DKL\pm o(1)) m_1) = n^{-c_1 \DKL \pm o(1)}.
  \end{align*}
  So
  \begin{align*}
    \bP(x\in I) = \sum_{k\in I} \bP(x=k) = n^{-c_1 \DKL \pm o(1)}.
  \end{align*}
  This proves \cref{eqn:lem-inf-binom-concentrate:i}.

  \cref{eqn:lem-inf-binom-concentrate:ii} follows from the Chernoff bound.
\end{proof}

Let $\bP^{(0)} = \Ber(1/2)^{\otimes n}$ denote the prior distribution of $A$ and $\bP^{(1)}$ denote the posterior distribution of $A$ conditioned on observations in Phase 1. Then for any set $B\subseteq [n]$ we have
\begin{align*}
  \bP^{(1)}(B) &\propto \bP\left((a_i)_{i\in [n]} | B\right) \bP^{(0)}(B) \\
  &= \left(\prod_{i\in B} p_{a_i}\right) \left(\prod_{i\in [n] \setminus B} p_{m_1-a_i}\right) \bP^{(0)}(B) \\
  &\propto \left(\prod_{i\in B} p_{a_i}\right) \left(\prod_{i\in [n] \setminus B} p_{m_1-a_i}\right).
\end{align*}

\subsection{Phase 2} \label{sec:inf:phase-2}
In Phase 2, the oracle reveals some elements in $A$ and not in $A$ as follows.
\begin{enumerate}[label=2\alph*.]
  \item In Step 2a, the oracle reveals elements $i$ with $a_i\not \in I$.
  \item In Step 2b, the oracle reveals every $i\in A$ independently with probability $q_{a_i}$, for some constants $(q_k)_{k\in I}$ to be chosen later.
\end{enumerate}

\subsubsection{Step 2a} \label{sec:inf:phase-2:step-a}
Let $S^{(2a)}_+$ (resp.~$S^{(2a)}_-$) denote the set of elements in $A$ (resp.~in $A^c$) revealed in Step 2a.
For $B\subseteq [n]$, we say $B$ is consistent with the observations $\left(S^{(2a)}_+, S^{(2a)}_-\right)$ if $S^{(2a)}_+ \subseteq B$ and $S^{(2a)}_-\cap B = \emptyset$.
Let $\bP^{(2a)}$ denote the posterior distribution at the end of Step 2a, and $\cC^{(2a)}$ be its support.
Then $B\subseteq [n]$ is in $\cC^{(2a)}$ if and only if it is consistent with $\left(S^{(2a)}_+, S^{(2a)}_-\right)$.
For $B\in \cC^{(2a)}$, the posterior probability $\bP^{(2a)}(B)$ is given by
\begin{align} \label{eqn:inf-phase-2-step-a:post}
  \bP^{(2a)}(B) \propto \bP^{(1)}(B)
  \propto \left(\prod_{i\in B} p_{a_i}\right) \left(\prod_{i\in B^c}p_{m_1 - a_i}\right).
\end{align}

\subsubsection{Step 2b} \label{sec:inf:phase-2:step-b}
Let $S^{(2b)}_+$ be the set of elements of $A$ revealed in Step 2b that were not revealed in Step 2a.
Define $S^{(\le 2b)}_+ = S^{(2a)}_+ \cup S^{(2b)}_+$.
Define $\bP^{(2b)}$, $\cC^{(2b)}$ similarly to \cref{sec:inf:phase-2:step-a}.
Then $\cC^{(2b)}$ is the set $B\in \cC^{(2a)}$ that are consistent with $\left(S^{(2b)}_+,\emptyset\right)$.
For any $B\in \cC^{(2b)}$, we have
\begin{align} \label{eqn:inf-phase-2-step-b:post-step}
  \bP^{(2b)}(B) &\propto \bP\left(S^{(2b)}_+ | B, S^{(2a)}_+\right) \bP^{(2a)}(B) \\
  \nonumber &\propto \left(\prod_{i\in B\backslash S^{(\le 2b)}_+} (1-q_{a_i})\right) \bP^{(2a)}(B) \\
  \nonumber &\propto \left(\prod_{i\in B\backslash S^{(\le 2b)}_+} (1-q_{a_i}) p_{a_i}\right)\left(\prod_{i\in B^c\backslash S^{(2a)}_-}p_{m_1 - a_i}\right) \\
  \nonumber &\propto \left(\prod_{i\in B\backslash S^{(\le 2b)}_+} \frac{(1-q_{a_i}) p_{a_i}}{p_{m_1 - a_i}}\right).
\end{align}
Note that for any $B\in \cC^{(2b)}$ and any $i\in B\backslash S^{(\le 2b)}_+$, we have $a_i \in I$.

Let us now choose the values of $q_k$ for $k\in I$.
For $k\in I$, define
\begin{align*}
  q_k = 1-\frac{p_{m_1-k} p_{k_l}}{p_k p_{m_1-k_l}}
\end{align*}
where $k_l=p m_1 - \log^{0.6}n$ is the left endpoint of $I$.
Because $\frac{p_{m_1-k}}{p_k} = \left(\frac{1-p}p\right)^{m_1-2k}$ is decreasing in $k$, we have $q_k\in [0, 1]$ for $k\in I$.
So this choice of $q_k$'s is valid.

With this choice of $q_k$, we can simplify \cref{eqn:inf-phase-2-step-b:post-step} as
\begin{align} \label{eqn:inf-phase-2-step-b:post}
  \bP^{(2b)}(B) \propto \left(\frac{p_{k_l}}{p_{m_1-k_l}}\right)^{|B|}.
\end{align}
\cref{eqn:inf-phase-2-step-b:post} is very useful and greatly simplifies the posterior distribution.
Importantly, $(a_i)_{i\in [n]}$ does not appear directly in the expression.
Therefore, $\left(S^{(\le 2b)}_+, S^{(2a)}_-, |A|\right)$ is a sufficient statistic for $A$ at the end of Step 2b.

Let us now consider the distribution of $\left(S^{(\le 2b)}_+, S^{(2a)}_-\right)$ conditioned on $A$.

Let $p_+$ be the probability that a coordinate $i\in A$ is in $S^{(\le 2b)}_+$.
Every coordinate $i\in A$ is independently in $S^{(2a)}_+$ with probability $\bP(\Bin(m_1,1-p)\not \in I) = 1-n^{-c_3\pm o(1)}$ (\cref{lem:inf:binom-concentrate}).
Every coordinate $i\in A$ is independently in $S^{(2b)}_+$ with probability
\begin{align*}
  \sum_{k\in I} p_k q_k = \sum_{k\in I} p_k \left(1-\frac{p_{m_1-k} p_{k_l}}{p_k p_{m_1-k_l}}\right).
\end{align*}
For a fixed $i\in A$, the events $i\in S^{(2a)}_+$ and $i\in S^{(2b)}_+$ are disjoint, so
\begin{align*}
  p_+ = \bP(\Bin(m_1,1-p)\not \in I) + \sum_{k\in I} p_k q_k \ge \bP(\Bin(m_1,1-p)\not \in I) =  1-n^{-c_3\pm o(1)}.
\end{align*}
On the other hand,
\begin{align*}
  p_+ = 1 - \sum_{k\in I} p_k (1 - q_k) \le 1 - p_{k_l} (1 - q_{k_l}) = 1 - p_{k_l} = 1-n^{-c_3\pm o(1)},
\end{align*}
where the last step is by \cref{lem:binomial-large-deviation}.
Thus,
\begin{align} \label{eqn:inf-p-plus}
    p_+ = 1-n^{-c_3\pm o(1)}.
\end{align}

Let $p_-$ be the probability that a coordinate $i\in A^c$ is in $S^{(2a)}_-$.
By \cref{lem:inf:binom-concentrate},
\begin{align} \label{eqn:inf-p-minus}
  p_- = \bP(\Bin(m_1,p)\not \in I) = o(1).
\end{align}

Therefore, observations up to Step 2b have the same effect as the following procedure:
\begin{definition}[Alternative observation procedure] \label{defn:inf:alt-obs}
  Let $A\sim \Ber(1/2)^{\otimes n}$.
  \begin{enumerate}[label=(\arabic*)]
    \item Observe every coordinate $i\in A$ independently with probability $p_+$ (\cref{eqn:inf-p-plus}).
    \item Observe every coordinate $i\in A^c$ independently with probability $p_-$ (\cref{eqn:inf-p-minus}).
  \end{enumerate}
\end{definition}

By \cref{eqn:inf-phase-2-step-b:post}, the posterior distribution of $A$ after Phase 2 is a biased product distribution on the unrevealed coordinates.

\subsection{Phase 3}
In Phase 3, the algorithm makes at most $c_2 n$ adaptive exact queries.
We will show that for $c_2$ small enough, no algorithm is able to determine $f(x)$ with very small error probability.

Our proof strategy is as follows.
\begin{enumerate}[label=3\alph*.]
  \item Because $f$ has linear total influence, after Phase 2 ends, the Boolean function on the unrevealed coordinates will have a biased version of total influence $\I_q$ at least $\Omega(n)$ in expectation.
  \item Every (adaptive) query made in Phase 3 decreases $\I_q$ by at most $1$. Therefore, after Phase 3, the Boolean function on the remaining unrevealed coordinates has $\I_q=\Omega(n)$ in expectation.
  \item Finally, we show that if a Boolean function has $\I_q=\Omega(n)$, then it is impossible to guess $f(x)$ (where $x$ follows a biased product distribution) with very small error probability.
\end{enumerate}

Because observations up to Phase 2 have been simplified by our analysis, we make some definitions and restate the problem we need to solve in Phase 3.

\subsubsection{Preliminaries}
We have a Boolean function $f: \{0,1\}^n \to \{0,1\}$ and an input $x\sim \Ber(1/2)^{\otimes n}$, and the goal is to determine $f(x)$.
Then we independently observe each $i\in [n]$ with probability is $p_+=1-n^{-c_3 \pm o(1)}$ for $x_i=1$ and $p_-=o(1)$ for $x_i=0$.
Let $s\in \{0,1,*\}^n$ be the observations. That is, if coordinate $i$ is revealed, then $s_i$ is the revealed value; otherwise $s_i=*$.
Let $U = \{i\in [n]: s_i=*\}$ be the unrevealed coordinates.
Conditioned on $s$, the distribution of $x_U$ is a product of $\Ber(q)$, where
\begin{align} \label{eqn:inf:phase-3:q}
  q=\frac{1-p_+}{1-p_+ + 1-p_-} = n^{-c_3 \pm o(1)}.
\end{align}
Let $f_s: \{0,1\}^U\to \{0,1\}$ be the function $f_s(x_U) = f(s\lhd x_U)$ for all $x_U \in \{0,1\}^U$, where $s\lhd x_U$ denotes the bit string where all $*$ in $s$ are replaced with the corresponding value in $x_U$.
Let $\rho$ be the distribution $\left(\frac{p_-}2, \frac{p_+}2, 1-\frac{p_-}2-\frac{p_+}2\right)$ on $\{0,1,*\}$.
Then without conditioning on $x$, $s$ has distribution $\rho^{\otimes n}$.

For any Boolean function $g: \{0,1\}^S\to \{0,1\}$, define the $q$-biased influence of coordinate $i\in S$ as
\begin{align*}
  \Inf_{q,i}(f) = \bP_{x\sim \Ber(q)^{\otimes S}}( f(x) \ne f(x\oplus e_i))
\end{align*}
and the $q$-biased total influence as
\begin{align*}
  \I_q(f) = \sum_{i\in S} \Inf_{q,i}(f).
\end{align*}
When we mention the $q$-biased total influence of the function $f_s$, the sum is over the unrevealed coordinates $i\in U$.

For a string $y\in \{0,1\}^S$, let $D_y$ denote the distribution of $t\in \{0,1,*\}^S$ where for $i\in S$ with $y_i=1$, $t_i=1$ with probability $p_+$ and $t_i=*$ with probability $1-p_+$; for $i\in S$ with $y_i=0$, $t_i=0$ with probability $p_-$ and $t_i=*$ with probability $1-p_-$.
For a string $t\in \{0,1,*\}^S$, let $E_s$ denote the distribution of $y\in \{0,1\}^S$ where $y_i=t_i$ if $t_i\in \{0,1\}$ and $y_i\sim \Ber(q)$ independently for $i\in S$ with $t_i=*$, where $q$ is defined in \cref{eqn:inf:phase-3:q}.
With these definitions, we have $x\sim \Ber(1/2)^{\otimes n}$, $s\sim \rho^{\otimes n}$, $x\sim E_s$ conditioned on $s$, $s\sim D_x$ conditioned on $x$.

\subsubsection{Step 3a}
We connect the total influence $\I(f)$ with the biased total influence $\I_q(f_s)$.

Fix $i\in [n]$, we have (in the following, $x \cup 0_i$ and $x \cup 1_i$ denotes setting the $i$-th bit of $x$ as $0$ and $1$, respectively)
\begin{align*}
  \Inf_i(f) &= \bP_{x\sim \Ber(1/2)^{\otimes ([n]\backslash i)}} ( f(x \cup 0_i) \ne f(x \cup 1_i)) \\
  &= \bP_{\substack{x\sim \Ber(1/2)^{\otimes ([n]\backslash i)}\\ s\sim D_x }} ( f(x \cup 0_i) \ne f(x \cup 1_i)) \\
  &= \bP_{\substack{s\sim \rho^{\otimes ([n]\backslash i)}\\ x\sim E_s }} ( f(x \cup 0_i) \ne f(x \cup 1_i)) \\
  &= \frac{1}{\bP_{s_i\sim \rho}(s_i=*)} \cdot \bP_{\substack{s\sim \rho^{\otimes [n]}\\ x\sim E_{s_{[n]\backslash i}} }} ( f(x \cup 0_i) \ne f(x \cup 1_i) \land s_i=*) \\
  &= \frac{1}{\bP_{s_i\sim \rho}(s_i=*)} \cdot \bP_{\substack{s\sim \rho^{\otimes [n]}\\ x\sim E_s }} ( f(x) \ne f(x \oplus e_i) \land s_i=*)\\
  &= \frac 1{\rho(*)} \cdot \bE_{s\sim \rho^{\otimes [n]}} \left[\mathbbm{1}_{s_i=*} \cdot \bP_{x\sim E_s} (f(x)\ne f(x \oplus e_i)) \right].
\end{align*}
The first step is by definition of $\Inf_i$.
The second step is because $s$ is not involved in the condition.
The third step is by considering the joint distribution between $x$ and $s$.
The fourth step is because the two conditions $f(x \cup 0_i) \ne f(x \cup 1_i)$ and $s_i=*$ are independent: the former depends on $s_{[n]\backslash i}$ and the latter depends on $s_i$.
The fifth step rewrites the condition. %
The sixth step changes the order of checking $s_i=*$ and choosing $x\sim E_s$.

Summing over $i\in [n]$, we have
\begin{align*}
  \rho(*) \I(f) &= \sum_{i\in [n]} \bE_{s\sim \rho^{\otimes [n]}} \left[\mathbbm{1}_{s_i=*} \cdot \bP_{x\sim E_s} (f(x)\ne f(x \oplus e_i)) \right] \\
  &= \bE_{s\sim \rho^{\otimes [n]}} \left[\sum_{i\in [n]: s_i=*} \bP_{x\sim E_s} (f(x)\ne f(x \oplus e_i)) \right] \\
  &= \bE_{s\sim \rho^{\otimes [n]}} \I_q(f_s).
\end{align*}
The second step is by linearity of expectation.
The third step is by definition of $\I_q$.

Note that $\rho(*) = 1-\frac{p_-}2-\frac{p_+}2 = \frac12 \pm o(1)$.
Because $\I(f) \ge c n$, we have
\begin{align} \label{eqn:inf:phase-3:step-1:q-inf}
  \bE_{s\sim \rho^{\otimes [n]}} \I_q(f_s) \ge (2c\pm o(1)) n.
\end{align}
That is, $f_s$ has expected $q$-biased total influence at least $(2c\pm o(1)) n$.

\subsubsection{Step 3b}
We prove the following lemma, which essentially says that adaptive exact queries in Phase 3 can only decrease the $q$-biased total influence $\I_q$ by a certain amount. In the following, recall the definition of $f_t$ for a Boolean function $f: \{0,1\}^n\to \{0,1\}$ and $t\in \{0,1,*\}^n$ is $f$ where we restrict all input coordinates $i$ with $t_i \ne *$ to be equal to $t_i$.

\begin{lemma} \label{lem:inf:phase-3:step-2}
  Let $f: \{0,1\}^n\to \{0,1\}$ be a Boolean function.
  Suppose the input $x$ follows distribution $\Ber(q)^{\otimes n}$.
  Consider an algorithm which adaptively makes at most $m$ exact queries in expectation.
  Let $t\in \{0,1,*\}^n$ be the random variable denoting the query results.
  Then $\bE[\I_q(f_t)] \ge \I_q(f)-m$, where $\bE$ is over the randomness of the revealed coordinates and the randomness of the algorithm.
\end{lemma}
\begin{proof}
  By induction it suffices to prove the case where the algorithm makes exactly one query.
  Without loss of generality, assume that the algorithm makes an query to coordinate $1$.
  Then $t=1*^{n-1}$ with probability $q$ and $t=0*^{n-1}$ with probability $1-q$.
  \begin{align*}
    \bE[\I_q(f_t)] &= q \I_q(f_{1*^{n-1}}) + (1-q) \I_q(f_{0*^{n-1}}) \\
    &= q \sum_{2\le i\le n} \Inf_{q,i}(f_{1*^{n-1}}) +(1-q) \sum_{2\le i\le n} \Inf_{q,i}(f_{0*^{n-1}}) \\
    &= \sum_{2\le i\le n} \Inf_{q,i}(f) \\
    &= \I_q(f) - \Inf_{q,1}(f) \\
    &\ge \I_q(f)-1.
  \end{align*}
  The first step is by expanding the expectation.
  The second step is by definition of $\I_q$.
  The third step is because $q \Inf_{q,i}(f_{1*^{n-1}}) + (1-q) \Inf_{q,i}(f_{0*^{n-1}}) = \Inf_{q,i}(f)$.
  The fourth step is by definition of $\I_q$.
  The fifth step is because $\Inf_{q,i}(f) \le 1$.
\end{proof}
We now apply \cref{lem:inf:phase-3:step-2} to Phase 3. Let $t\in \{0,1,*\}^U$ be the observations made in Phase 3, where $U = \{i\in [n]: s_i=*\}$ is the set of unrevealed coordinates at the end of Phase 2.
Then \cref{eqn:inf:phase-3:step-1:q-inf} together with \cref{lem:inf:phase-3:step-2} implies that
\begin{align} \label{eqn:inf:phase-3:step-2:q-inf}
  \bE_{s\sim \rho^{\otimes [n]}} [ \bE_t [\I_q((f_{s})_{t})] ]\ge (2c - c_2\pm o(1)) n.
\end{align}

\subsubsection{Step 3c}
After Phase 3, the Boolean function on the unrevealed coordinates has $q$-biased total influence at least $(2c-c_2\pm o(1))n = \Omega(n)$ in expectation.
In particular, with probability $\ge c-c_2/2\pm o(1)$, the function has $I_q((f_{s})_t) \ge (c-c_2/2\pm o(1))n$.
Any algorithm now needs to output an answer in $\{0,1\}$.
The following result shows that the error probability will be $\Omega(1)$ no matter what the algorithm outputs.

\begin{lemma} \label{lem:inf:phase-3:step-3}
  For any $0 < c \le 1$, there exists $c'>0$ such that the following holds.
  Let $f: \{0,1\}^n\to \{0,1\}$ be a Boolean function and $\frac 1{3n} \le q \le c/6$ be a parameter.
  If $\I_q(f) \ge cn$, then \[c' \le \bE_{x \sim \Ber(q)^{\otimes n}} [f(x)] \le 1 - c'.\]
\end{lemma}
\begin{proof}
    Since $\I_q(f) \ge cn$ and $\Inf_{q,i}(f) \le 1$ for every $i$, there exist at least $cn/2$ indices $i\in [n]$ with $\Inf_{q,i}(f) \ge c / 2$. As $\lfloor c / 6q\rfloor \le cn/2$, there are $m = \lfloor c / 6q\rfloor$ indices $i\in [n]$ with $\Inf_{q,i}(f) \ge c / 2$. Without loss of generality, assume that $\Inf_{q,i}(f) \ge c / 2$ for $i\in [m]$.

    Let $0^m$ denote the length-$m$ bit string with all $0$'s. By union bound, we have
    \begin{align}
    \label{lem:inf:phase-3:step-3:all-0}
    \begin{split}
    \bP_{x \sim \Ber(q)^{\otimes n}}\left(x_{[m]}=0^m\right) \ge 1 - \sum_{i\in [m]} \bP_{x \sim \Ber(q)^{\otimes n}}\left(x_i = 1\right) = 1 - mq \ge 1 - c / 6.
    \end{split}
    \end{align}

    Let $\diamond$ denote the concatenation operation of two bit strings. For any $i \in [m]$, we have
    \begin{align}
    \nonumber &~\bP_{y \sim \Ber(q)^{\otimes ([n]\backslash [m])}}\left(f(0^{m} \diamond y) \ne f(e_i \diamond y)\right) \\
    \nonumber =&~\bP_{x \sim \Ber(q)^{\otimes n}}\left(f(x) \ne f(x \oplus e_i) \mid x_{[m]}=0^m\right)\\
    \nonumber \ge&~\bP_{x \sim \Ber(q)^{\otimes n}}\left(f(x) \ne f(x \oplus e_i) \wedge x_{[m]}=0^m\right)\\
    \nonumber \ge &~\bP_{x \sim \Ber(q)^{\otimes n}}\left(f(x) \ne f(x \oplus e_i)\right) + \bP_{x \sim \Ber(q)^{\otimes n}}\left(x_{[m]}=0^m\right) - 1 \\
    \nonumber \ge & \Inf_{q, i}(f) +(1 - c/6) - 1 \tag{By \cref{lem:inf:phase-3:step-3:all-0}}\\
    \ge &~c / 3. \label{eq:lem:inf:phase-3:step-3:eq2}
    \end{align}

    Finally, we have
    \begin{align*}
    &\bE_{x \sim \Ber(q)^{\otimes n}}(f(x))\\
    &\ge \bP_{x \sim \Ber(q)^{\otimes n}}(f(x) = 1 \wedge x_1 + \cdots + x_m = 1)\\
    &= \sum_{i\in [m]} \bP_{x \sim \Ber(q)^{\otimes n}}(f(x) = 1 \wedge x_{[m]} = e_i)\\
    &= q (1-q)^{m-1} \cdot \sum_{i\in [m]} \bP_{y \sim \Ber(q)^{\otimes ([n]\backslash [m])}}(f(e_i \diamond y) = 1)\\
    &\ge  \Theta_c(q) \cdot \sum_{i\in [m]} \bP_{y \sim \Ber(q)^{\otimes ([n]\backslash [m])}}(f(0^m \diamond y) = 0 \wedge f(0^m \diamond y) \ne f(e_i \diamond y))\\
    &\ge  \Theta_c(q) \cdot \sum_{i\in [m]} \left( \bP_{y \sim \Ber(q)^{\otimes ([n]\backslash [m])}}(f(0^m \diamond y) = 0) +  \bP_{y \sim \Ber(q)^{\otimes ([n]\backslash [m])}} (f(0^m \diamond y) \ne f(e_i \diamond y)) - 1\right)\\
    &\ge  \Theta_c(1) \cdot \left( \bP_{y \sim \Ber(q)^{\otimes ([n]\backslash [m])}}(f(0^m \diamond y) = 0) +  c / 3 - 1\right) \tag{By \cref{eq:lem:inf:phase-3:step-3:eq2}}\\
    &\ge  \Theta_c(1) \cdot \left( \bP_{x \sim \Ber(q)^{\otimes n}}(f(x) = 0 \wedge x_{[m]}=0^m) +  c / 3 - 1\right)\\
    &\ge  \Theta_c(1) \cdot \left( \bP_{x \sim \Ber(q)^{\otimes n}}(f(x) = 0) + (1 - c / 6) - 1 +  c / 3 - 1\right) \tag{By \cref{lem:inf:phase-3:step-3:all-0}}\\
    &= \Theta_c(1) \cdot (c/6 -  \bE_{x \sim \Ber(q)^{\otimes n}}(f(x))).
    \end{align*}
    Therefore, there exists $c'$ depending only on $c$ such that $\bE_{x \sim \Ber(q)^{\otimes n}}[f(x)] \ge c'$.

    By symmetry, we also have $\bE_{x \sim \Ber(q)^{\otimes n}}(f(x)) \le 1 - c'$.
\end{proof}

Applying \cref{lem:inf:phase-3:step-3} to the Boolean function $(f_s)_t$ on the unrevealed coordinates finishes the proof.

\section{Graph Connectivity}
Recall the Graph Connectivity problem, where the input is an unknown undirected graph on $n$ labeled vertices. In each query, the algorithm picks an unordered pair $e=(u,v)\in \binom V2$, and the oracle returns whether $e$ is an edge of $G$, flipped independently with probability $0<p<\frac 12$. The goal of the algorithm is to determine whether $G$ is connected or not.

In this section, we prove \cref{thm:graph-conn-hard}, which we recall below:
\GraphConnectivity*

\subsection{Preliminaries} \label{sec:conn:prelim}

\begin{lemma}[Stirling's formula] \label{lem:stirling}\
  \begin{enumerate}
    \item There exists an absolute constant $\epsilon>0$ such that for all $n\ge 1$,
    \begin{align*}
      \epsilon \sqrt n\left(\frac ne\right)^n \le n! \le \epsilon^{-1} \sqrt n\left(\frac ne\right)^n.
    \end{align*}
    \item Fix $m\ge 2$. There exists an constant $\epsilon=\epsilon(m)>0$ such that for all $n\ge 1$ and all $k_1,\ldots,k_m\ge 1$ with $k_1+\cdots+k_m=n$,
    \begin{align*}
      \epsilon \cdot \frac{n^{n+1/2}}{\prod_{i\in [m]} k_i^{k_i+1/2}} \le \binom{n}{k_1,\ldots,k_m} \le \epsilon^{-1} \cdot \frac{n^{n+1/2}}{\prod_{i\in [m]} k_i^{k_i+1/2}}.
    \end{align*}
  \end{enumerate}
\end{lemma}

Let $(\bP_n)_n$ and $(\bQ_n)_n$ be two sequences of probability measures such that $\bP_n$ and $\bQ_n$ are defined on the same measurable spaces $\Omega_n$.
We say $\bQ_n$ is contiguous with respect to $\bP_n$, denoted by $\bQ_n\contig \bP_n$, if for every sequence $(\cA_n)_n$ of measurable sets $\cA_n\subseteq \Omega_n$, $\bP_n(\cA_n) \to 0$ implies $\bQ_n(\cA_n)\to 0$.

\begin{lemma}[Cayley's formula] \label{lem:cayley}
The number of spanning trees on $n$ labeled vertices is $n^{n-2}$.
\end{lemma}

\subsection{Structure of a uniform spanning tree} \label{sec:conn:ust}
Our proof of \cref{thm:graph-conn-hard} uses several properties of the uniform random spanning tree (UST) of the complete graph.
In this section we state and prove these properties.

The notion of balanced edges is crucial to our construction and analysis.
\begin{definition}[Balanced edges] \label{defn:conn:balanced-edge}
  Let $\beta>0$ be a constant.
  Let $T$ be a spanning tree on $n$ labeled vertices.
  An edge $e\in T$ is called \emph{$\beta$-balanced} if both sides of $e$ has at least $\beta n$ vertices.
  Let $B_\beta(T)$ denote the set of $\beta$-balanced edges of $T$.
\end{definition}

\begin{proposition} \label{prop:conn:ust-structure-1}
  Let $T$ be a UST on $n$ labeled vertices.
  There exist absolute constants $\epsilon,\gamma_1,\gamma_2,\gamma_3>0$ such that with probability at least $\epsilon$, the following are true simultaneously.
  \begin{enumerate}[label=(\roman*)]
    \item\label{item:prop:conn:ust-structure-1:i}
    \begin{align*}
      \gamma_1 \sqrt n \le \left| B_{1/3}(T) \right| \le \gamma_2 \sqrt n.
    \end{align*}
    \item\label{item:prop:conn:ust-structure-1:ii} For all $e\in B_{1/3}(T)$, if $T_1$, $T_2$ are the two connected components of $T\backslash e$, then
    \begin{align*}
      \min\{\left| B_{1/7}(T_1) \right|,\left| B_{1/7}(T_2) \right|\} \ge \gamma_3 \sqrt n.
    \end{align*}
  \end{enumerate}
\end{proposition}
Our proof of \cref{prop:conn:ust-structure-1} uses the following lemmas.

\begin{lemma}[Balanced edges form a chain] \label{lem:conn:balanced-edge-chain}
  Let $T$ be a tree on $n$ labeled vertices.
  Then the subgraph formed by all edges in $B_{1/3}(T)$ is either empty or a chain.
\end{lemma}
\begin{proof}
  Suppose $B_{1/3}(T)$ is non-empty.
  Let $H$ be the subgraph formed by all edges in $B_{1/3}(T)$.

  \paragraph{Step 1.} We prove that $H$ is connected.
  Let $e_1,e_2\in B_{1/3}(T)$ and $e_3$ be an edge on the path (in $T$) between $e_1$ and $e_2$.
  Let $T_{i,1},T_{i,2}$ ($i=1,2,3$) be the two connected components of $T\backslash e_i$,
  such that $e_3\in T_{1,1}$, $e_3\in T_{2,1}$, $e_1\in T_{3,1}$, $e_2\in T_{3,2}$.
  Then $|T_{i,j}|\ge \frac n3$ for $i,j\in \{1,2\}$.
  We have $T_{3,i} \supseteq T_{i,2}$ ($i=1,2$).
  So $|T_{3,i}| \ge \frac n3$ for $i=1,2$ and $e_3$ is $\frac 13$-balanced.
  This shows that $H$ is connected.

  \paragraph{Step 2.} We prove that $H$ is a chain.
  By Step 1, $H$ is a connected subgraph of $T$, so it is a tree.
  Suppose $H$ has a vertex $v$ of degree at least three.
  Then there exist three distinct $\frac 13$-balanced edges $e_1,e_2,e_3$ containing $v$.
  Say $e_i=(v,u_i)$ ($i=1,2,3$).
  Let $T_i$ ($i=1,2,3$) be the connected component of $T\backslash e_i$ not containing $v$.
  Then $|T_i|\ge \frac n3$ and $T_i$ ($i=1,2,3$) are all disjoint.
  This implies
  \begin{align*}
    n\ge |T_1|+|T_2|+|T_3|+1 \ge n+1,
  \end{align*}
  which is a contradiction.
  So all vertices of $H$ have degree at most two, and $H$ must be a chain.
\end{proof}

The following statement about typical distances in a UST is well-known (e.g., \cite{aldous1991continuum}).
\begin{lemma}[Typical distance in UST] \label{lem:conn:ust-dist}
  Let $T$ be a UST on $n$ labeled vertices and $u,v$ be two fixed vertices (not dependent on $T$).
  Then for any $\epsilon>0$, there exist absolute constants $\gamma_1,\gamma_2>0$, such that
  \begin{align*}
    \bP\left[ \gamma_1 \sqrt n \le \dist_T(u,v) \le \gamma_2 \sqrt n \right] \ge 1-\epsilon,
  \end{align*}
  where $\dist_T$ denotes the graph distance in $T$.
\end{lemma}

To prove \cref{prop:conn:ust-structure-1}, we define a probability measure $\bQ$ on the space of spanning trees on $n$ labeled vertices that is contiguous with respect to the UST measure, and \cref{prop:conn:ust-structure-1} holds under $\bQ$.

\begin{definition}[Measure $\bQ$] \label{defn:conn:ust-measure-q}
  We define the measure $\bQ$ as follows.
  A \emph{tree tuple} $\cD$ is a tuple which consists of the following data:
  \begin{enumerate}[label=(\roman*)]
    \item Six distinct vertices $u_0,u_1,u_2,v_0,v_1,v_2\in [n]$.
    \item Integers $\frac n3-\frac n{100}\le L_0,R_0 \le \frac n3-1$.
    Integers $\frac n3-L_0\le L_1 \le \frac{n}{100}$, $\frac n3-R_0\le R_1 \le \frac{n}{100}$.
    Let $W = n-L_0-L_1-R_0-R_1$.
    \item A partition of $V=[n]$ into five subsets $V=V_{L_0}\sqcup V_{L_1}\sqcup V_{R_0}\sqcup V_{R_1} \sqcup V_W$, where $u_i\in V_{L_i}$, $v_i\in V_{R_i}$, $u_2,v_2\in V_W$, and $|V_{L_i}|=L_i$, $|V_{R_i}|=R_i$, $|V_W|=W$ ($i=1,2$).
    \item Five spanning trees: $T_{L_i}$ on $V_{L_i}$, $T_{R_i}$ on $V_{R_i}$, $T_W$ on $V_W$ ($i=1,2$).
  \end{enumerate}
  Given a tree tuple $\cD$, it produces a spanning tree
  \begin{align*}
    T(\cD) = T_{L_0} \cup T_{L_1} \cup T_{R_0} \cup T_{R_1} \cup T_W \cup \{(u_0,u_1),(u_1,u_2),(v_0,v_1),(v_1,v_2)\}.
  \end{align*}
  The measure $\bQ$ is the distribution of $T(\cD)$, where $\cD$ is uniformly chosen from all tree tuples.
\end{definition}

\begin{lemma}[Contiguity] \label{lem:conn:ust-contig}
  Let $\bP$ be the UST measure on $n$ labeled vertices.
  Let $\bQ$ be as defined in \cref{defn:conn:ust-measure-q}.
  Then $\bQ \contig \bP$.
\end{lemma}
\begin{proof}
  Let us study the support of $\bQ$.
  By \cref{lem:conn:balanced-edge-chain}, $\frac 13$-balanced edges in $T$ form a chain.
  By the construction (\cref{defn:conn:ust-measure-q}), for $T=T(\cD)$, $u_1$ and $v_1$ are the two endpoints of the chain $B_{1/3}(T)$, and $u_2$ (resp.~$v_2$) is the second vertex on the path from $u_1$ to $v_1$ (resp.~$v_1$ to $u_1$).
  Furthermore, $u_0$ (resp.~$v_0$) is the only neighbor $x$ of $u_1$ (resp.~$v_1$) not on the path between $u_1$ and $v_1$, such that $u_0$'s connected component in $T\backslash (u_0,x)$ (resp.~$v_0$'s connected component in $T\backslash (v_0,x)$) has size at least $\frac n3-\frac{n}{100}$.
  This means that $\cD$ can be reconstructed given $T(\cD)$, up to swapping $(u_0,u_1,u_2)$ and other data with $(v_0,v_1,v_2)$ and the corresponding data.
  Specifically, every $T$ in the support of $\bQ$ is realized by exactly two $\cD$s.
  Therefore, $\bQ$ is equal to the conditional measure $\bP(\cdot \mid T\in \supp \bQ)$.

  Now, to prove that $\bQ \contig \bP$, it suffices to prove that $\left| \supp \bQ \right| = \Theta(\left| \supp \bP \right|) = \Theta(n^{n-2})$.
  Because every $T$ is realized exactly twice, it suffices to prove that the number of different tree tuples is $\Theta(n^{n-2})$.
  We construct a tree tuple according to the following procedure.
  \begin{enumerate}[label=(\roman*)]
    \item Choose six distinct vertices $u_0,u_1,u_2,v_0,v_1,v_2\in [n]$.
    There are $n!/(n-6)!=\Theta(n^6)$ ways to do this.
    \item Choose $L_1$ and $R_1$ from $[1, \frac{n}{100}]$.
    Choose $L_0$ from $\left[\frac n3-L_1, \frac n3-1\right]$ and $R_0$ from $\left[\frac n3-R_1, \frac n3-1\right]$.
    \item Choose the partition $V=V_{L_0}\sqcup V_{L_1}\sqcup V_{R_0}\sqcup V_{R_1} \sqcup V_W$.
    There are
    \begin{align*}
      \binom{n-6}{L_0-1,L_1-1,R_0-1,R_1-1,W-2}
    \end{align*}
    ways to do this.
    \item Choose the five spanning trees, $T_{L_0}$, $T_{L_1}$, $T_{R_0}$, $T_{R_1}$, $T_W$.
    There are
    \begin{align*}
      L_0^{L_0-2} L_1^{L_1-2} R_0^{R_0-2} R_1^{R_1-2} W^{W-2}
    \end{align*}
    ways to do this.
  \end{enumerate}
  Summarizing the above, for large enough $n$, the number of different tree tuples is
  \begin{align*}
    &~\frac{n!}{(n-6)!} \cdot \sum_{1\le L_1,R_1 \le \frac{n}{100}} \sum_{\substack{\frac n3-L_1 \le L_0 \le \frac n3-1\\ \frac n3-R_1 \le R_0 \le \frac n3-1}} \binom{n-6}{L_0-1,L_1-1,R_0-1,R_1-1,W-2} \\
    \nonumber &~\cdot L_0^{L_0-2} L_1^{L_1-2} R_0^{R_0-2} R_1^{R_1-2} W^{W-2} \\
    \nonumber \asymp &~ \sum_{1\le L_1,R_1 \le \frac{n}{100}} \sum_{\substack{\frac n3-L_1 \le L_0 \le \frac n3-1\\ \frac n3-R_1 \le R_0 \le \frac n3-1}} \binom{n}{L_0,L_1,R_0,R_1,W} L_0^{L_0-1} L_1^{L_1-1} R_0^{R_0-1} R_1^{R_1-1} W^W \\
    \nonumber \asymp &~ \sum_{1\le L_1,R_1 \le \frac{n}{100}} \sum_{\substack{\frac n3-L_1 \le L_0 \le \frac n3-1\\ \frac n3-R_1 \le R_0 \le \frac n3-1}} \frac{n^{n+1/2}}{(L_0 L_1 R_0 R_1)^{3/2} W^{1/2}} \\
    \nonumber \asymp &~ \sum_{1\le L_1,R_1 \le \frac{n}{100}} \sum_{\substack{\frac n3-L_1 \le L_0 \le \frac n3-1\\ \frac n3-R_1 \le R_0 \le \frac n3-1}} \frac{n^{n-3}}{(L_1 R_1)^{3/2}} \\
    \nonumber \asymp &~ \sum_{1\le L_1,R_1 \le \frac{n}{100}} \frac{n^{n-3}}{(L_1 R_1)^{1/2}} \\
    \nonumber \asymp &~ n^{n-2}.
  \end{align*}
  This finishes the proof.
\end{proof}

\begin{lemma}[\cref{prop:conn:ust-structure-1}, \cref{item:prop:conn:ust-structure-1:i} for $\bQ$] \label{lem:conn:ust-structure-q-i}
  Let $T$ be sampled from $\bQ$ (\cref{defn:conn:ust-measure-q}).
  For any $\epsilon>0$, there exist $\gamma_1,\gamma_2>0$ such that with probability at least $1-\epsilon$, \cref{prop:conn:ust-structure-1}, \cref{item:prop:conn:ust-structure-1:i} holds for $T$.
\end{lemma}
\begin{proof}
  As we discussed in the proof of \cref{lem:conn:ust-contig}, all edges on the path between $u_1$ and $v_1$ are $\frac 13$-balanced.
  Conditioned on $V_W$, $T_W$ being a UST on $V_W$,
  by \cref{lem:conn:ust-dist},
  \begin{align*}
    \bP\left[ \gamma_1 \sqrt W \le \dist_{T_W}(u_1,v_1) \le \gamma_2 \sqrt W \right] \ge 1-\epsilon.
  \end{align*}
  Noting that $\dist_T(u_1,v_1)=\dist_{T_W}(u_1,v_1)$ and $W = \Theta(n)$, we finish the proof.
\end{proof}

\begin{corollary}[\cref{prop:conn:ust-structure-1}, \cref{item:prop:conn:ust-structure-1:i}] \label{coro:conn:ust-structure-i}
  Let $T$ be a UST on $n$ labeled vertices.
  There exist absolute constants $\epsilon,\gamma_1,\gamma_2>0$ such that with probability at least $\epsilon$, \cref{prop:conn:ust-structure-1}, \cref{item:prop:conn:ust-structure-1:i} holds for $T$.
\end{corollary}
\begin{proof}
  By \cref{lem:conn:ust-structure-q-i,lem:conn:ust-contig}.
\end{proof}
Using \cref{coro:conn:ust-structure-i}, we can prove that \cref{prop:conn:ust-structure-1} holds for $\bQ$.
\begin{lemma}[\cref{prop:conn:ust-structure-1} for $\bQ$] \label{lem:conn:ust-structure-q-all}
  Let $T$ be sampled from $\bQ$ (\cref{defn:conn:ust-measure-q}).
  There exist absolute constants $\epsilon,\gamma_1,\gamma_2,\gamma_3>0$ such that with probability at least $\epsilon$, both items in \cref{prop:conn:ust-structure-1} hold for $T$.
\end{lemma}
\begin{proof}
  Let $\cE_{L_0}$ (resp.~$\cE_{R_0}$, $\cE_W$) be the event that $\left| B_{1/3}(T_{L_0}) \right|$ (resp.~$\left| B_{1/3}(T_{R_0}) \right|$, $\dist_{T_W}(u_1,v_1)$) is in $[\gamma_1 \sqrt n, \gamma_2 \sqrt n]$.
  Let $\cE = \cE_{L_0} \cap \cE_{R_0} \cap \cE_W$.
  Let $\cA$ be the $\sigma$-algebra generated by $V_{L_0}$, $V_{R_0}$, $V_W$.
  Conditioned on $\cA$, the three subtrees $T_{L_0}$, $T_{R_0}$, $T_W$ are independent and are USTs on the respective vertex sets.
  By \cref{lem:conn:ust-dist}, \cref{coro:conn:ust-structure-i}, and $L_0,R_0,W=\Theta(n)$, there exist $\epsilon,\gamma_1,\gamma_2>0$ such that
  \begin{align*}
    \bP(\cE_{L_0} |\cA), \bP(\cE_{R_0} |\cA), \bP(\cE_W |\cA) \ge \epsilon^{1/3}.
  \end{align*}
  By independence, we have $\bP(\cE | \cA) \ge \epsilon$.
  Therefore $\bP(\cE) \ge \epsilon$.

  Now we show that when $\cE$ happens, both items in \cref{prop:conn:ust-structure-1} hold.
  As we discussed in the proof of \cref{lem:conn:ust-contig}, $B_{1/3}(T)$ is the set of edges on the path between $u_1$ and $v_1$.
  So \cref{item:prop:conn:ust-structure-1:i} holds because $\cE_W$ happens.
  Now consider \cref{item:prop:conn:ust-structure-1:ii}.
  Let $e$ be an edge on the path between $u_1$ and $v_1$.
  Let $T_1$, $T_2$ be the two components of $T\backslash e$, where $u_1\in T_1$, $v_1\in T_2$.
  Then $T_{L_0} \subseteq T_1$, $T_{R_0} \subseteq T_2$.
  Every $\frac 13$-balanced edge of $T_{L_0}$ divides $T_{L_0}$ into components of size at least $\frac{|T_{L_0}|}3$,
  so it is a $\beta$-balanced edge of $T_1$ with $\beta = \frac{|T_{L_0}|}{3 |T_1|} \ge \frac{\frac{n}{3} - \frac{n}{100}}{3 \cdot \frac{2}{3} n} \ge \frac 17$.
  The same discussion holds for $T_{R_0}$ and $T_2$.
  Because $\cE_{L_0}$ and $\cE_{R_0}$ happen, \cref{item:prop:conn:ust-structure-1:ii} holds.
\end{proof}

By \cref{lem:conn:ust-structure-q-all,lem:conn:ust-contig}, we complete the proof of \cref{prop:conn:ust-structure-1}.

\begin{lemma} \label{lem:conn:ust-structure-2}
  Let $T$ be a UST. Then for any $\epsilon>0$, $0<\beta<\frac 12$, there exists $\gamma>0$ such that
  \begin{align*}
    \bP\left[ \left| B_\beta(T) \right| \le \gamma \sqrt n \right] \ge 1-\epsilon.
  \end{align*}
\end{lemma}
\begin{proof}
  The proof is by the first moment method.
  Let $e=(u,v)$ be an unordered pair of vertices.
  For any integer $\beta n \le k \le (1-\beta) n$, the number of spanning trees containing $e$ such that $u$'s side contains exactly $k$ (including $u$) vertices is
  \begin{align*}
    \binom{n-2}{k-1} k^{k-2} (n-k)^{n-k-2}
    \asymp n^{-2} \binom nk k^{k-1} (n-k)^{n-k-1}
    \asymp \frac{n^{n-3/2}}{(k(n-k))^{3/2}}
    \asymp n^{n-9/2}.
  \end{align*}
  So the number of trees $T$ with $e\in B_\beta(T)$ is
  \begin{align*}
    \asymp \sum_{k\in [\beta n, (1-\beta) n]} n^{n-9/2} \asymp n^{n-7/2}.
  \end{align*}
  Therefore,
  \begin{align*}
    \bE \left[B_\beta(T)\right] \asymp \frac{\binom n2 \cdot n^{n-7/2}}{n^{n-2}} \asymp \sqrt n.
  \end{align*}
  The result then follows from Markov's inequality.
\end{proof}

\begin{lemma} \label{lem:conn:ust-structure-3}
  Let $T$ be a UST.
  For $e\in T$, let $s_T(e)$ denote the size of the smaller component of $T\backslash e$.
  Then for any $\epsilon>0$, there exists $\gamma>0$ such that
  \begin{align*}
    \bP\left[ \sum_{e\in T} s_T(e) \le \gamma n^{3/2} \right] \ge 1-\epsilon.
  \end{align*}
\end{lemma}
\begin{proof}
  The proof is by first moment method.
  Let $e=(u,v)$ be an unordered pair of vertices.
  We define $s_T(e)=0$ if $e\not \in T$.
  Then
  \begin{align*}
    \bE \left[s_T(e)\right] &\asymp \frac 1{n^{n-2}} \cdot \sum_{1\le k\le \frac n2} k \cdot \binom{n-2}{k-1} k^{k-2} (n-k)^{n-k-2} \\
    \nonumber &\asymp \frac 1{n^n} \sum_{1\le k\le \frac n2} \binom nk k^k (n-k)^{n-k-1} \\
    \nonumber &\asymp \sum_{1\le k\le \frac n2} \frac{n^{1/2}}{k^{1/2} (n-k)^{3/2}} \\
    \nonumber &\asymp \sum_{1\le k\le \frac n2} \frac{1}{k^{1/2} n} \\
    \nonumber &\asymp n^{-1/2}.
  \end{align*}
  Therefore,
  \begin{align*}
    \bE\left[ \sum_{e\in T} s_T(e)\right] \asymp \binom n2 \cdot n^{-1/2} \asymp n^{3/2}.
  \end{align*}
  The result then follows from Markov's inequality.
\end{proof}

Combining \cref{prop:conn:ust-structure-1}, \cref{lem:conn:ust-structure-2} and \cref{lem:conn:ust-structure-3}, we get the following corollary.
\begin{corollary} \label{coro:conn:ust-structure-final}
  Let $T$ be a UST on $n$ labeled vertices.
  There exist absolute constants $\epsilon,\gamma_1,\gamma_2,\gamma_3,\gamma_4,\gamma_5>0$ such that the following are true simultaneously.
  \begin{enumerate}[label=(\roman*)]
    \item\label{item:coro:conn:ust-structure-final:i}
    \begin{align*}
      \gamma_1\sqrt n \le \left| B_{1/3}(T) \right| \le \left| B_{1/42}(T) \right| \le \gamma_2 \sqrt n.
    \end{align*}
    \item\label{item:coro:conn:ust-structure-final:ii} For all $e\in B_{1/3}(T)$, if $T_1$ and $T_2$ are the two connected components of $T\backslash e$, then
    \begin{align*}
      \gamma_3 \sqrt n \le \min\{\left| B_{1/7}(T_1) \right|,\left| B_{1/7}(T_2) \right|\} \le \max\{\left| B_{1/14}(T_1) \right|,\left| B_{1/14}(T_2) \right|\} \le \gamma_4 \sqrt n.
    \end{align*}
    \item\label{item:coro:conn:ust-structure-final:iii}
    \begin{align*}
      \sum_{e\in T} s_T(e) \le \gamma_5 n^{3/2}.
    \end{align*}
  \end{enumerate}
\end{corollary}
\begin{proof}
  By \cref{prop:conn:ust-structure-1}, there exist $\epsilon,\gamma_1,\gamma_3>0$ such that with probability at least $\epsilon$, the lower bounds in \cref{coro:conn:ust-structure-final}, \cref{item:coro:conn:ust-structure-final:i,item:coro:conn:ust-structure-final:ii} hold.
  By \cref{lem:conn:ust-structure-2}, there exists $\gamma_2>0$ such that
  \begin{align*}
    \bP\left[\left| B_{1/42}(T) \right| \le \gamma_2 \sqrt n\right] \ge 1-\epsilon/3.
  \end{align*}
  By \cref{lem:conn:ust-structure-3}, there exists $\gamma_5>0$ such that
  \begin{align*}
    \bP\left[\sum_{e\in T} s_T(e) \le \gamma_5 n^{3/2} \right] \ge 1-\epsilon/3.
  \end{align*}
  By union bound, with probability at least $\epsilon/3$, \cref{coro:conn:ust-structure-final}, \cref{item:coro:conn:ust-structure-final:i,item:coro:conn:ust-structure-final:iii}, and the lower bound in \cref{coro:conn:ust-structure-final}, \cref{item:coro:conn:ust-structure-final:ii} hold.
  Now notice that $B_{1/14}(T_1) \cup B_{1/14}(T_2) \subseteq B_{1/42}(T)$, so the upper bound in \cref{coro:conn:ust-structure-final}, \cref{item:coro:conn:ust-structure-final:ii} holds with $\gamma_4=\gamma_2$.
\end{proof}

\subsection{Hard distribution and three-phase problem} \label{sec:conn:init}
We now start the proof of \cref{thm:graph-conn-hard}.
Our input distribution is defined as follows.
\begin{definition}[Hard distribution for graph connectivity] \label{defn:conn:hard-dist}
  Let $\beta_0=\frac 1{21}$.
  We generate the input graph $G$ using the following procedure.
  \begin{enumerate}[label=(\arabic*)]
    \item Let $T$ be a UST on vertex set $V=[n]$.
    \item If $B_{\beta_0}(T)=\emptyset$, return to the previous step.
    Otherwise, let $e_0\sim \Unif(B_{\beta_0}(T))$.
    \item Throw a fair coin $z\sim \Ber(1/2)$.
    If $z=1$, output $G=T$; if $z=0$, output $G=T\backslash e_0$.
  \end{enumerate}
\end{definition}

We define the following three-phase problem, where the algorithm has more power than in the noisy query model.
\begin{definition}[Three-phase problem] \label{defn:conn:three-phase}
  Let $G$ be generated from \cref{defn:conn:hard-dist}.
  Let $c_1,c_2>0$ be constants to be determined later.
  Consider the following three-phase problem.
  \begin{enumerate}[label=\arabic*.]
    \item \label{item:sec:conn:phase-1} Phase 1: The algorithm makes $m_1=c_1\log n$ noisy queries to every unordered pair of vertices $(u,v)\in \binom V2$.
    \item \label{item:sec:conn:phase-2} Phase 2: The oracle reveals some edges and non-edges of $G$. The choice of these edges and non-edges will be described later.
    \item \label{item:sec:conn:phase-3} Phase 3: The algorithm makes up to $m_2=c_2 n^2$ (adaptive) exact queries.
  \end{enumerate}
  The goal of the algorithm is to determine whether the input graph $G$ is connected.
\end{definition}

\begin{lemma} \label{lem:conn:reduction}
  If no algorithm can solve the three-phase problem (\cref{defn:conn:three-phase}) with error probability $\epsilon>0$, then no algorithm can solve the graph connectivity problem with error probability $\epsilon$ using at most $m_1 m_2 = c_1 c_2 n^2 \log n$ noisy queries.
\end{lemma}
The proof is similar to \cref{lem:inf:reduction} and omitted.

\begin{proposition}[Hardness of the three-phase problem] \label{prop:conn:three-phase-hard}
  For some choices of $c_1$, $c_2$, and Phase 2 strategy (\cref{defn:conn:three-phase}), the following is true:
  there exists $\epsilon>0$ such that no algorithm can solve the three-phase problem (\cref{defn:conn:three-phase}) with error probability $\epsilon$.
\end{proposition}

\begin{proof}[Proof of \cref{thm:graph-conn-hard}]
  Combining \cref{prop:conn:three-phase-hard} and \cref{lem:conn:reduction}.
\end{proof}

The following sections are devoted to the proof of \cref{prop:conn:three-phase-hard}.

\subsection{Phase 1} \label{sec:conn:phase-1}
In Phase 1, the algorithm makes $m_1=c_1 \log n$ noisy queries to every potential edge $e\in \binom V2$.
Let $a_e$ denote the number of times where a query to $e$ returns $1$.
Then for $e\in G$, $a_e\sim \Bin(m_1, 1-p)$; for $e\not \in G$, $a_e \sim \Bin(m_1, p)$.
For $0\le k\le m_1$, define
\begin{align*}
  p_k = \bP(\Bin(m_1, 1-p) = k) = \binom{m_1}{k} (1-p)^k p^{m_1 - k}.
\end{align*}

Let $I = \left[p m_1 - \log^{0.6}n, p m_1 + \log^{0.6}n\right]$.
\begin{lemma} \label{lem:conn:binom-concentrate}
  Let $x\sim \Bin(m_1, 1-p)$, $y\sim \Bin(m_1, p)$.
  Then
  \begin{align}
    \label{eqn:lem-conn-binom-concentrate:i} \bP(x\in I) &= n^{-c_3 \pm o(1)}, \\
    \label{eqn:lem-conn-binom-concentrate:ii} \bP(y\in I) &= 1-o(1),
  \end{align}
  where $c_3 = c_1 (1-2p) \log \frac{1-p}p$.
\end{lemma}
The proof is the same as \cref{lem:inf:binom-concentrate} and omitted.

Let $\bP^{(0)}$ denote the graph distribution in \cref{defn:conn:hard-dist}.
Let $\bP^{(1)}$ denote the posterior distribution of $G$ conditioned on observations in Phase 1.
Let $\cC^{(0)}$ (resp.~$\cC^{(1)}$) deonte the support of $\bP^{(0)}$ (resp.~$\bP^{(1)}$).
Then $\cC^{(1)} = \cC^{(0)}$ and for any graph $H \in \cC^{(0)}$, we have
\begin{align} \label{eqn:conn-phase-1:post}
  \bP^{(1)}(H) &\propto \bP\left((a_e)_{e\in \binom V2}| H\right) \bP^{(0)}(H),\\
  \nonumber &= \left(\prod_{e\in H} p_{a_e}\right)\left(\prod_{e\in H^c} p_{m_1 - a_e}\right) \bP^{(0)}(H).
\end{align}
where $H^c$ denotes the complement $\binom V2\backslash H$.

\subsection{Phase 2} \label{sec:conn:phase-2}
In Phase 2, the oracle reveals some edges and non-edges of $G$ as follows.
\begin{enumerate}[label=2\alph*.]
  \item In Step 2a, the oracle reveals potential edges $e$ with $a_e\not \in I$.
  \item In Step 2b, the oracle reveals every $e\in G$ independently with probability $q_{a_e}$. We choose $q_j = 1-\frac{p_{m_1-j} p_{j_l}}{p_j p_{m_1-j_l}}$ for $j\in I$ where $j_l = p m_1 - m_1^{0.6}$.
  \item In Step 2c, the oracle reveals $n-2$ edges of $G$ as follows. If $G$ is disconnected, reveal all edges of $G$. Otherwise, $G$ is connected and is some tree $T$. If $G$ has a $\beta_0$-balanced edge that is not revealed yet, uniformly randomly choose an edge $e^*$ from all such edges, and reveal all edges of $G\backslash e^*$. If all $\beta_0$-balanced edges of $G$ have been revealed, report failure.
\end{enumerate}

\subsubsection{Step 2a and Step 2b} \label{sec:conn:phase-2:step-ab}

By the same analysis as in \cref{sec:inf:phase-2:step-a,sec:inf:phase-2:step-b}, observations up to Step 2b have the same effect as the following procedure:
\begin{definition}[Alternative observation procedure] \label{defn:conn:alt-obs}
  Let $G$ be  generated as in \cref{defn:conn:hard-dist}.
  \begin{enumerate}[label=(\arabic*)]
    \item Observe every edge $e\in G$ independently with probability
    \begin{align*}
      p_+ &= 1-\sum_{j\in I} p_j (1-q_j)
      = 1 - \frac{p_{j_l}}{p_{m_1-j_l}} \cdot \sum_{j\in I} p_{m_1-j} \\
      \nonumber &= 1-(1\pm o(1)) \frac{p_{j_l}}{p_{m_1-j_l}} = 1-n^{-c_3\pm o(1)}.
    \end{align*}
    \item Observe every non-edge $e\in G^c$ independently with probability
    \begin{align*}
      p_- = \bP(\Bin(m_1,p)\not \in I) = o(1).
    \end{align*}
  \end{enumerate}
\end{definition}

Let $\bP^{(2b)}$ be the posterior distribution of $G$ after Step 2b and $\cC^{(2b)}$ be its support.
Let $E^{(2a)}_+$ (resp.~$E^{(2a)}_-$) denote the set of edges (resp.~non-edges) revealed in Step 2a.
Let $E^{(2b)}_+$ be the set of edges revealed in Step 2b that were not revealed in Step 2a.
Define $E^{(\le 2b)}_+ = E^{(2a)}_+ \cup E^{(2b)}_+$.
Then $\cC^{(2b)}$ is the set of graphs $H\in \cC^{(0)}$ satisfying $E^{(\le 2b)}_+ \subseteq H$ and $E^{(2a)}_-\cap H = \emptyset$.

By the same analysis as in \cref{sec:inf:phase-2:step-a,sec:inf:phase-2:step-b}, the posterior distribution satisfies
\begin{align} \label{eqn:conn-phase-2-step-b:post}
  \bP^{(2b)}(H) \propto \left(\frac{p_{k_l}}{p_{m_1-k_l}}\right)^{|H|-(n-2)} \bP^{(0)}(H).
\end{align}
for $H\in \cC^{(2b)}$.

\subsubsection{Step 2c} \label{sec:conn:phase-2:step-c}
Let $E^{(2c)}_+$ be the set of edges revealed in Step 2c that were not revealed in previous steps.
Let $\bP^{(2c)}$ be the posterior distribution of $G$ and $\cC^{(2c)}$ be the support of $\bP^{(2c)}$.

\begin{lemma} \label{lem:conn:phase-2:step-c:failure}
  Conditioned on $G$ being connected, with probability $\Omega(1)$, Step 2c does not report failure, and $e^*$ is $\frac 13$-balanced.
\end{lemma}
\begin{proof}
  Step 2c reports failure when $G$ is connected and all $\beta_0$-balanced edges have been revealed in previous steps.
  Let $T$ be a UST.
  By \cref{coro:conn:ust-structure-final}, for some $\epsilon,\gamma_1,\gamma_2>0$, we have
  \begin{align*}
    \bP\left[\gamma_1 \sqrt n \le \left| B_{1/3}(T) \right| \le \left| B_{\beta_0}(T) \right| \le \gamma_2 \sqrt n\right] > \epsilon.
  \end{align*}
  When constructing the input distribution (\cref{defn:conn:balanced-edge}), conditioned on $G$ being connected, the distribution of $G$ is uniform over all spanning trees with at least one $\beta_0$-balanced edge.
  Let $\cE_1$ be the event that $G$ is connected, has at least $\gamma_1 \sqrt n$ $\frac 13$-balanced edges, and at most $\gamma_2 \sqrt n$ $\beta_0$-balanced edges.
  By the above discussion, $\cP(\cE_1) = \Omega(1)$.
  In \cref{defn:conn:alt-obs}, every edge is observed independently with probability $p_+ = 1-n^{-c_3\pm o(1)}$.
  By choosing $c_1>0$ small enough, we can let $c_3>0$ be arbitrarily small.
  Let $\wt B_{1/3}$ (resp.~$\wt B_{\beta_0}$) be the set of $\frac 13$-balanced (resp.~$\beta_0$-balanced) edges of $G$ not revealed in previous steps.
  Then
  \begin{align*}
    \wt B_{1/3} \subseteq \wt B_{\beta_0}, \qquad
    \left|\wt B_{1/3}\right| \sim \Bin\left(\left| B_{1/3}(G) \right|, 1-p_+\right), \qquad
    \left|\wt B_{\beta_0}\right| \sim \Bin\left(\left| B_{\beta_0}(G) \right|, 1-p_+\right).
  \end{align*}
  Note that
  \begin{align*}
    \bE\left[\Bin\left(\left| B_{1/3}(G) \right|, 1-p_+\right)\right] &= \left| B_{1/3}(G) \right| \cdot (1-p_+) = n^{1/2-c_3 \pm o(1)}, \\
    \bE\left[\Bin\left(\left| B_{\beta_0}(G) \right|, 1-p_+\right)\right] &= \left| B_{\beta_0}(G) \right| \cdot (1-p_+) = n^{1/2-c_3 \pm o(1)}.
  \end{align*}
  By Bernstein's inequality, there exists $\gamma_3,\gamma_4>0$ such that conditioned on $\cE_1$, we have
  \begin{align*}
    \bP\left[ \left|\wt B_{1/3}\right| \ge \gamma_3 \left| B_{1/3}(G) \right| \cdot (1-p_+) \mid \cE_1\right] \ge 1-\exp\left(-n^{1/2-c_3 \pm o(1)}\right),\\
    \bP\left[ \left|\wt B_{\beta_0}\right| \le \gamma_4 \left| B_{\beta_0}(G) \right| \cdot (1-p_+) \mid \cE_1\right] \ge 1-\exp\left(-n^{1/2-c_3 \pm o(1)}\right).
  \end{align*}
  Let $\cE_2$ be the event that $\left|\wt B_{1/3}\right| \ge \gamma_3 \left| B_{1/3}(G) \right| \cdot (1-p_+)$ and $\left|\wt B_{\beta_0}\right| \le \gamma_4 \left| B_{\beta_0}(G) \right| \cdot (1-p_+)$.
  The above discussion shows that $\bP(\cE_2 | \cE_1) = 1-o(1)$.

  When $\cE_1$ and $\cE_2$ both happen, we have
  $
    \frac{\left|\wt B_{1/3}\right|}{\left|\wt B_{\beta_0}\right|} \ge \frac{\gamma_3 \gamma_1}{\gamma_4 \gamma_2} > 0
  $.
  So conditioned on $\cE_1\cap \cE_2$, the probability that $e^*$ is $\frac 13$-balanced is $\Omega(1)$.
  This finishes the proof.
\end{proof}

From now on we condition on the event that Step 2c does not report failure, and $e^*$ is $\frac 13$-balanced.
Let us consider the posterior distribution $\bP^{(2c)}$.
Let $E^{(\le 2c)}_+ = E^{(\le 2b)}_+ \cup E^{(2c)}_+$ be the set of observed edges at the end of Step 2c.
Then $E^{(\le 2c)}_+$ consists of $n-2$ edges, which is a forest with two components $T_1$ and $T_2$, each containing at least $\frac n3$ and at most $\frac{2n}3$ vertices.
The support $\cC^{(2c)}$ is easy to describe.
Let $G_0 = T_1 \cup T_2$ and $G_e = T_1 \cup T_2 \cup \{e\}$ for $e\in E(T_1, T_2)$.
Then
\begin{align*}
  \cC^{(2c)} = \left\{G_0\right\} \cup \left\{G_e : e\in E(T_1,T_2)\setminus E^{(2a)}_-\right\}.
\end{align*}
The posterior distribution $\bP^{(2c)}$ is not simply the distribution $\bP^{(2b)}$ restricted to $\cC^{(2c)}$.
For $H\in \cC^{(2c)}$, we have
\begin{align} \label{eqn:sec:conn:phase-2:step-c:post-step}
  \bP^{(2c)}(H) &\propto \bP\left( E^{(2c)}_+ \mid H, E^{(\le 2b)}_+\right) \bP^{(2b)}(H) \\
  \nonumber &\propto \bP\left( E^{(2c)}_+ \mid H, E^{(\le 2b)}_+\right) \left(\frac{p_{k_l}}{p_{m_1-k_l}}\right)^{|H|-(n-2)} \bP^{(0)}(H).
\end{align}
For $H=G_0$, \cref{eqn:sec:conn:phase-2:step-c:post-step} simplifies to
\begin{align} \label{eqn:sec:conn:phase-2:step-c:post-g0}
  \bP^{(2c)}(G_0) \propto \bP^{(0)}(G_0).
\end{align}
For $H=G_e$, \cref{eqn:sec:conn:phase-2:step-c:post-step} simplifies to
\begin{align} \label{eqn:sec:conn:phase-2:step-c:post-ge}
  \bP^{(2c)}(G_e) \propto \frac 1{\left| B_{\beta_0}(G_e) \backslash E^{(\le 2b)}_+ \right|} \cdot \frac{p_{k_l}}{p_{m_1-k_l}} \bP^{(0)}(G_e).
\end{align}
Note that the $\propto$ symbols in \cref{eqn:sec:conn:phase-2:step-c:post-g0,eqn:sec:conn:phase-2:step-c:post-ge} hide the same factor.

Further simplifying \cref{eqn:sec:conn:phase-2:step-c:post-g0,eqn:sec:conn:phase-2:step-c:post-ge}, we get
\begin{align}
  \label{eqn:sec:conn:phase-2:step-c:post-g0-2} \bP^{(2c)}(G_0) &\propto \sum_{e\in E(T_1,T_2)} \frac 1{\left| B_{\beta_0}(G_e) \right|}, \\
  \label{eqn:sec:conn:phase-2:step-c:post-ge-2} \bP^{(2c)}(G_e) &\propto \frac 1{\left| B_{\beta_0}(G_e) \backslash E^{(\le 2b)}_+ \right|} \cdot \frac{p_{k_l}}{p_{m_1-k_l}}, \qquad \forall e\in E(T_1,T_2) \backslash E^{(2a)}_-.
\end{align}

We now consider the set $B_{\beta_0}(G_e)$ for $e\in E(T_1,T_2)$.
Let $\beta_1 = \frac{\beta_0 n}{|T_1|}$, $\beta_2 = \frac{\beta_0 n}{|T_2|}$.
Then $\frac 1{14}\le \beta_1,\beta_2\le \frac 17$ and $B_{\beta_1}(T_1) \cup B_{\beta_2}(T_2) \subseteq B_{\beta_0}(G_e)$ for all $e\in E(T_1,T_2)$.
For $e_i\in T_i$ ($i=1,2$), let $S_{T_i}(e_i)$ be the set of vertices in the smaller component of $T_i\backslash e_i$. (If the two components have the same size, choose a side arbitrarily.)
For $e=(u_1,u_2)\in E(T_1,T_2)$ (with $u_i\in T_i$, $i=1,2$), an edge $e_i\in T_i\backslash B_{\beta_i}(T_i)$ ($i=1,2$) is in $B_{\beta_0}(G_e)$ if and only if $u_i\in S_{T_i}(e_i)$.
For $i\in \{1,2\}$ and $e=(u_1,u_2)\in E(T_1,T_2)$, define
\begin{align*}
  B'_{T_i}(e) = \{e_i : e_i\in T_i\backslash B_{\beta_i}(T_i), u_i\in S_{T_i}(e_i)\}.
\end{align*}
Then for $e\in E(T_1,T_2)$, we have
\begin{align} \label{eqn:sec:conn:phase-2:step-c:new-balanced}
  B_{\beta_0}(G_e) = \{e\} \cup B_{\beta_1}(T_1) \cup B_{\beta_2}(T_2) \cup B'_{T_1}(e) \cup B'_{T_2}(e).
\end{align}
Note that the union is a disjoint union.

\begin{lemma} \label{lem:conn:phase-2:step-c:structure-1}
  Conditioned on $G$ being connected, there exist constants $\epsilon,\gamma_1,\gamma_2,\gamma_3>0$ such that with probability at least $\epsilon$, the following are true simultaneously.
  \begin{enumerate}[label=(\roman*)]
    \item\label{item:lem:conn:phase-2:step-c:structure-1:i} Step 2c does not report failure and $e^*\in B_{1/3}(G)$.
    \item\label{item:lem:conn:phase-2:step-c:structure-1:ii}
    \begin{align*}
      \gamma_1 \sqrt n \le \left| B_{\beta_1}(T_1) \right| + \left| B_{\beta_2}(T_2) \right| \le \gamma_2 \sqrt n.
    \end{align*}
    \item\label{item:lem:conn:phase-2:step-c:structure-1:iii}
    \begin{align*}
      \sum_{e_1\in T_1\backslash B_{\beta_1}(T_1)} \left| S_{T_1}(e_1) \right| + \sum_{e_2\in T_2\backslash B_{\beta_2}(T_2)} \left| S_{T_2}(e_2) \right| \le \gamma_3 n^{3/2}.
    \end{align*}
  \end{enumerate}
\end{lemma}
\begin{proof}
  Let $T$ be a UST.
  Let $\cE_1$ be the event that all items in \cref{coro:conn:ust-structure-final} hold.
  Then conditioned on $G$ being connected, $\cE_1$ happens with probability $\Omega(1)$.
  In the following, condition on that $G$ is connected and $\cE_1$ happens.

  Let $\cE_2$ be the event that \cref{lem:conn:phase-2:step-c:structure-1}, \cref{item:lem:conn:phase-2:step-c:structure-1:i} holds.
  By the proof of \cref{lem:conn:phase-2:step-c:failure}, conditioned on $\cE_1$, $\cE_2$ happens with probability $\Omega(1)$.
  In the following, condition on that $\cE_1$ and $\cE_2$ both happen.

  By \cref{coro:conn:ust-structure-final}, \cref{item:coro:conn:ust-structure-final:ii}, and because $e^*\in B_{1/3}(T)$,
  \begin{align*}
    \left| B_{\beta_1}(T_1) \right| + \left| B_{\beta_2}(T_2) \right|&\ge \left| B_{1/7}(T_1) \right| + \left| B_{1/7}(T_2) \right| \ge \gamma_3 \sqrt n,\\
    \left| B_{\beta_1}(T_1) \right| + \left| B_{\beta_2}(T_2) \right|&\le \left| B_{1/14}(T_1) \right| + \left| B_{1/14}(T_2) \right| \le \gamma_4 \sqrt n.
  \end{align*}
  Therefore conditioned on $\cE_1$ and $\cE_2$, \cref{lem:conn:phase-2:step-c:structure-1}, \cref{item:lem:conn:phase-2:step-c:structure-1:ii} holds.

  By \cref{coro:conn:ust-structure-final}, \cref{item:coro:conn:ust-structure-final:iii}, we have
  \begin{align*}
    \sum_{e_1\in T_1\backslash B_{\beta_1}(T_1)} \left| S_{T_1}(e_1) \right| + \sum_{e_2\in T_2\backslash B_{\beta_2}(T_2)} \left| S_{T_2}(e_2) \right|
    \le \sum_{e\in T} \left| S_T(e) \right| \le \gamma_5 n^{3/2}.
  \end{align*}
  Therefore conditioned on $\cE_1$ and $\cE_2$, \cref{lem:conn:phase-2:step-c:structure-1}, \cref{item:lem:conn:phase-2:step-c:structure-1:iii} holds.
\end{proof}

\begin{corollary} \label{coro:conn:phase-2:step-c:structure-2}
  Conditioned on $G$ being connected, there exist constants $\epsilon,\gamma_1,\gamma_2,\gamma_3,\gamma_4,\gamma_5>0$ such that with probability at least $\epsilon$, the following are true simultaneously.
  \begin{enumerate}[label=(\roman*)]
    \item\label{item:coro:conn:phase-2:step-c:structure-2:i} For all $e\in E(T_1,T_2)$,
    \begin{align*}
      \left| B_{\beta_0}(G_e) \right| \ge \gamma_1 \sqrt n.
    \end{align*}
    \item\label{item:coro:conn:phase-2:step-c:structure-2:ii} For all $e\in E(T_1,T_2)$,
    \begin{align} \label{eqn:item:coro:conn:phase-2:step-c:structure-2:ii}
      \gamma_2 \left| B_{\beta_0}(G_e) \right| \cdot (1-p_+)\le \left| B_{\beta_0}(G_e) \backslash E^{(\le 2b)}_+ \right| \le \gamma_3 \left| B_{\beta_0}(G_e) \right| \cdot (1-p_+).
    \end{align}
    \item\label{item:coro:conn:phase-2:step-c:structure-2:iii}
    \begin{align*}
      \gamma_4 n^{3/2} \le \sum_{e\in E(T_1,T_2)} \frac 1{\left| B_{\beta_0}(G_e) \right|} \le \gamma_5 n^{3/2}.
    \end{align*}
  \end{enumerate}
\end{corollary}
\begin{proof}
  Let $T$ be a UST. Let $\cE$ be the event that all items in \cref{lem:conn:phase-2:step-c:structure-1} hold.
  Then conditioned $G$ being connected, $\cE$ happens with probability $\Omega(1)$.
  In the following, condition on that $G$ is connected and $\cE$ happens.

  By \cref{lem:conn:phase-2:step-c:structure-1}, \cref{item:lem:conn:phase-2:step-c:structure-1:ii} and \cref{eqn:sec:conn:phase-2:step-c:new-balanced}, we have
  \begin{align*}
    \left| B_{\beta_0}(G_e) \right| \ge \left| B_{\beta_1}(T_1) \right| + \left| B_{\beta_2}(T_2) \right| \ge \gamma_1 \sqrt n.
  \end{align*}
  So \cref{coro:conn:phase-2:step-c:structure-2}, \cref{item:coro:conn:phase-2:step-c:structure-2:i} holds.
  This implies the upper bound in \cref{coro:conn:phase-2:step-c:structure-2}, \cref{item:coro:conn:phase-2:step-c:structure-2:iii} as
  \begin{align*}
    \sum_{e\in E(T_1,T_2)} \frac 1{\left| B_{\beta_0}(G_e) \right|} \le n^2 \cdot \frac 1{\gamma_1 \sqrt n} = \gamma_1^{-1} n^{3/2}.
  \end{align*}

  By \cref{coro:conn:phase-2:step-c:structure-2}, \cref{item:coro:conn:phase-2:step-c:structure-2:i} and Bernstein's inequality, for every $e\in E(T_1,T_2)$, with probability $1-\exp\left(-n^{1/2-c_3\pm o(1)}\right)$, \cref{eqn:item:coro:conn:phase-2:step-c:structure-2:ii} holds.
  By union bound, with probability $1-o(1)$, \cref{eqn:item:coro:conn:phase-2:step-c:structure-2:ii} holds for all $e\in E(T_1,T_2)$.
  This proves \cref{coro:conn:phase-2:step-c:structure-2}, \cref{item:coro:conn:phase-2:step-c:structure-2:ii}.

  It remains to prove the lower bound in \cref{coro:conn:phase-2:step-c:structure-2}, \cref{item:coro:conn:phase-2:step-c:structure-2:iii}.
  By \cref{lem:conn:phase-2:step-c:structure-1}, \cref{item:lem:conn:phase-2:step-c:structure-1:ii,item:lem:conn:phase-2:step-c:structure-1:iii}, and \cref{eqn:sec:conn:phase-2:step-c:new-balanced}, we have
  \begin{align*}
    &~\sum_{e\in E(T_1,T_2)} \left| B_{\beta_0}(G_e) \right| \\
    \nonumber =&~ \sum_{e\in E(T_1,T_2)} \left(1 + \left| B_{\beta_1}(T_1) \right| + \left| B_{\beta_2}(T_2) \right| + \left| B'_{T_1}(e) \right| + \left| B'_{T_2}(e) \right|\right) \\
    \nonumber \le&~ \sum_{e\in E(T_1,T_2)} \left( 1 + \gamma_2 \sqrt n + \left| B'_{T_1}(e) \right| + \left| B'_{T_2}(e) \right|\right) \\
    \nonumber \le&~ (1+\gamma_2) \sqrt n |T_1||T_2| + \sum_{e_1\in T_1\backslash B_{\beta_1}(T_1)} \left| S_{T_1}(e_1) \right| \cdot |T_2|
    + \sum_{e_2\in T_2\backslash B_{\beta_2}(T_2)} \left| S_{T_2}(e_2) \right| \cdot |T_1| \\
    \nonumber \le&~ (1 + \gamma_2+\gamma_3) n^{5/2}.
  \end{align*}
  On the other hand, by Cauchy-Schwarz inequality,
  \begin{align*}
    \left( \sum_{e\in E(T_1,T_2)} \frac 1{\left| B_{\beta_0}(G_e) \right|} \right)\left( \sum_{e\in E(T_1,T_2)} \left| B_{\beta_0}(G_e) \right| \right)
    \ge (|T_1||T_2|)^2 \ge \frac 29 n^2.
  \end{align*}
  So
  \begin{align*}
    \sum_{e\in E(T_1,T_2)} \frac 1{\left| B_{\beta_0}(G_e) \right|} \ge \frac 2{9(\gamma_2+\gamma_3)} n^{3/2}.
  \end{align*}
\end{proof}

Now we are able to further utilize \cref{eqn:sec:conn:phase-2:step-c:post-g0-2} and \cref{eqn:sec:conn:phase-2:step-c:post-ge-2}.
Write
\begin{align*}
  Z^{(2c)} = \sum_{e\in E(T_1,T_2)} \frac 1{\left| B_{\beta_0}(G_e) \right|} + \sum_{e\in E(T_1,T_2) \backslash E^{(2a)}_-} \frac 1{\left| B_{\beta_0}(G_e) \backslash E^{(\le 2b)}_+ \right|} \cdot \frac{p_{k_l}}{p_{m_1-k_l}}.
\end{align*}
Conditioned on that all items in \cref{coro:conn:phase-2:step-c:structure-2} hold, for all $e\in E(T_1,T_2) \backslash E^{(2a)}_-$, we have
\begin{align*}
  \frac 1{\left|B_{\beta_0}(G_e) \backslash E^{(\le 2b)}_+\right|} \cdot \frac{p_{k_l}}{p_{m_1-k_l}}
  \asymp \frac 1{\left|B_{\beta_0}(G_e)\right|(1-p_+)} \cdot \frac{p_{k_l}}{p_{m_1-k_l}}
  \asymp\frac 1{\left|B_{\beta_0}(G_e)\right|},
\end{align*}
where the second step holds because
\begin{align*}
  \frac{p_{m_1-k_l}}{p_{k_l}}\cdot (1-p_+)
  = \frac{p_{m_1-k_l}}{p_{k_l}} \cdot \sum_{k\in I} p_k(1-q_k)
  = \sum_{k\in I} p_{m_1-k} = 1-p_- = 1-o(1).
\end{align*}
Therefore,
\begin{align*}
  Z^{(2c)} &\asymp \sum_{e\in E(T_1,T_2)} \frac 1{\left| B_{\beta_0}(G_e) \right|} \asymp n^{3/2},\\
  \bP^{(2c)}(G_0) &= \frac 1{Z^{(2c)}} \sum_{e\in E(T_1,T_2)} \frac 1{\left| B_{\beta_0}(G_e) \right|} \asymp 1.
\end{align*}

By Bernstein's inequality, for any $\delta_1>0$, we have
\begin{align*}
  \bP\left[\left| E^{(2a)}_- \right| \ge \delta_1 n^2\right] = o(1).
\end{align*}
So for small enough $\delta_1>0$, with probability $1-o(1)$, we have
\begin{align*}
  &~\sum_{e\in E(T_1,T_2) \backslash E^{(2a)}_-} \bP^{(2c)}(G_e) \\
  \nonumber =&~ \frac 1{Z^{(2c)}} \sum_{e\in E(T_1,T_2) \backslash E^{(2a)}_-} \frac 1{\left|B_{\beta_0}(G_e) \backslash E^{(\le 2b)}_+\right|} \cdot \frac{p_{k_l}}{p_{m_1-k_l}} \\
  \nonumber \asymp &~ n^{-3/2} \sum_{e\in E(T_1,T_2) \backslash E^{(2a)}_-} \frac 1{\left|B_{\beta_0}(G_e)\right|} \\
  \nonumber \asymp &~ n^{-3/2} \left(\sum_{e\in E(T_1,T_2)} \frac 1{\left|B_{\beta_0}(G_e)\right|} - \sum_{e\in E^{(2a)}_-} \frac 1{\left|B_{\beta_0}(G_e)\right|} \right) \\
  \nonumber \asymp &~ n^{-3/2} \left(n^{3/2} - \delta_1 n^2 \cdot n^{-1/2} \right) \\
  \nonumber \asymp &~ 1,
\end{align*}
where in the second-to-last step we used \cref{coro:conn:phase-2:step-c:structure-2}, \cref{item:coro:conn:phase-2:step-c:structure-2:i,item:coro:conn:phase-2:step-c:structure-2:iii}, and $\left| E^{(2a)}_- \right| \le \delta_1 n^2$.

Summarizing the above, at the end of Step 2c, with probability $\Omega(1)$, we have
\begin{align*}
  \bP^{(2c)}(\text{disconnected}) &= \bP^{(2c)}(G_0) = \Theta(1), \\
  \bP^{(2c)}(\text{connected}) &= \sum_{e\in E(T_1,T_2) \backslash E^{(2a)}_-} \bP^{(2c)}(G_e) = \Theta(1).
\end{align*}

\subsection{Phase 3} \label{sec:conn:phase-3}
In Phase 3, the algorithm makes $c_2 n^2$ adaptive exact queries.
We show that for $c_2>0$ small enough, with probability $\Omega(1)$, the algorithm will not be able to return the correct answer.

Let $E^{(3)}$ be the set of edges queried in Phase 3.
We can w.l.o.g.~assume that $E^{(3)} \subseteq E(T_1,T_2) \backslash E^{(2a)}_-$, because only queries in this set are useful.
Conditioned on $G$ being connected, the probability that $E^{(3)}$ hits the edge $e^*$ is
\begin{align*}
  \frac{\sum_{e\in E^{(3)}} \bP^{(2c)}(G_e)}{\sum_{e\in E(T_1,T_2) \backslash E^{(2a)}_-} \bP^{(2c)}(G_e)}
  \le c_2 n^2 \cdot \Theta(n^{-2}) \asymp c_2,
\end{align*}
which is $1-\Omega(1)$ for $c_2>0$ small enough.
Therefore, for small enough $c_2$, with probability $\Omega(1)$, $E^{(3)}$ does not hit the edge $e^*$.

Conditioned on $E^{(3)}$ does not hit $e^*$, let $\bP^{(3)}$ denote the posterior distribution of the original graph $G$ given all observations.
We have
\begin{align*}
  \bP^{(3)}(\text{disconnected}) &= \bP^{(3)}(G_0),\\
  \bP^{(3)}(\text{connected}) &= \sum_{e\in E(T_1,T_2) \backslash \left(E^{(2a)}_-\cup E^{(3)}\right)} \bP^{(3)}(G_e),
\end{align*}
where
\begin{align*}
  \bP^{(2c)}(G_0) &=\frac 1{Z^{(3)}} \sum_{e\in E(T_1,T_2)} \frac 1{\left| B_{\beta_0}(G_e) \right|}, \\
  \bP^{(2c)}(G_e) &=\frac 1{Z^{(3)}} \frac 1{\left| B_{\beta_0}(G_e) \backslash E^{(\le 2b)}_+ \right|} \cdot \frac{p_{k_l}}{p_{m_1-k_l}},\\
  Z^{(3)} &= \sum_{e\in E(T_1,T_2)} \frac 1{\left| B_{\beta_0}(G_e) \right|} + \sum_{e\in E(T_1,T_2) \backslash \left(E^{(2a)}_-\cup E^{(3)}\right)} \frac 1{\left| B_{\beta_0}(G_e) \backslash E^{(\le 2b)}_+ \right|} \cdot \frac{p_{k_l}}{p_{m_1-k_l}}.
\end{align*}

By the same discussion as in the end of Phase 2, Step 2c, for $c_2>0$ small enough, with probability $\Omega(1)$, we have
\begin{align*}
  \bP^{(3)}(\text{disconnected}) &= \Theta(1), \\
  \bP^{(3)}(\text{connected}) &= \Theta(1).
\end{align*}
In this case, any return value would lead to an error probability of $\Omega(1)$.

This concludes the proof of \cref{prop:conn:three-phase-hard}.

\subsection{\texorpdfstring{$s$-$t$}{s-t} Connectivity} \label{sec:conn:s-t}
In this section we modify the proof of \cref{thm:graph-conn-hard} to show hardness of $s$-$t$ Connectivity.
Recall the $s$-$t$ Connectivity problem, where the input is an unknown undirected graph on $n$ labeled vertices, and a pair of vertices $s,t \in V$. An algorithm can make noisy queries to edge membership and the goal is to determine whether $s$ and $t$ are in the same connected component of $G$.

\begin{proposition}[Hardness of $s$-$t$ Connectivity] \label{prop:s-t-conn-hard}
  Any algorithm that solves the $s$-$t$ Connectivity problem with $\frac 13$ error probability uses $\Omega(n^2\log n)$ noisy queries in expectation.
\end{proposition}

\begin{proof}
  As discussed in \cref{sec:conn:overview}, the error probability in the proposition statement can be replaced with any $0<\epsilon<\frac 12$, and the expected number of queries can be replaced with worst-case number of queries.

  We design an input distribution for $s$-$t$ Connectivity by generating $G$ from \cref{defn:conn:hard-dist}, and choosing $s,t$ i.i.d.~$\sim \Unif(V)$.

  Then we run the same proof as \cref{thm:graph-conn-hard}.
  That is, we define a three-phase problem for $s$-$t$ Connectivity, where the oracle uses the same strategy in Phase 2 as in Graph Connectivity.
  Because $s$ and $t$ are independent with $G$, in the end of Phase 2, conditioned on Step 2c does not report failure, with probability $\Omega(1)$, $s\in T_1$ and $t\in T_2$.
  In this case, $s$-$t$ Connectivity is equivalent to Graph Connectivity.
  In the proof of Graph Connectivity, we have shown that with probability $\Omega(1)$ (over the randomness of the graph and Phase 1 and 2), any algorithm that uses at most $c_2 n^2$ queries in Phase 3 has $\Omega(1)$ error probability.
  This implies that the same holds for $s$-$t$ Connectivity.
\end{proof}

\section{Threshold and Counting}
In this section we present our proof for \cref{thm:k-threshold} and \cref{thm:counting}.
We use $\Threshold{n}{k}$ to denote $k$-Threshold problem with input length $n$.

\subsection{Lower bound for \texorpdfstring{$k$}{k}-Threshold}
In this section we prove the lower bound part of \cref{thm:k-threshold}. That is, solving $\Threshold{n}{k}$ for $k\le 2n-1$ requires at least $(1-o(1)) \frac{n \log \frac k\delta}{\DKL}$ noisy queries in expectation.
The $k=o(n)$ case has been proved in \cite{wang2024noisy} (see also \cref{sec:th-small} for our alternative and simpler proof).

\begin{restatable}{theorem}{ThresholdSmallK}\label{thm:th:small-k}
For $k = o(n)$, solving $\Threshold{n}{k}$ with $
\delta = o(1)$ error probability requires
\[
(1-o(1))\frac{n \log \frac{k}{\delta}}{\DKL}
\]
noisy queries in expectation,
even when the input is uniformly chosen from $\binom{[n]}{k}$ with probability $1/2$ and uniformly chosen from $\binom{[n]}{k-1}$ with probability $1/2$. %
\end{restatable}
In our lower bound proof, we will use \cref{thm:th:small-k} with $k=O(n/\log n)$.

We first prove the case where $k=(n+1)/2$.
\begin{lemma}
\label{lem:th:maj}
  Solving $\Threshold{2k-1}{k}$ with $
\delta = o(1)$ error probability requires
\[
(1-o(1))\frac{2k \log \frac{k}{\delta}}{\DKL}
\]
noisy queries in expectation, even when the input is uniformly chosen from $\binom{[2k-1]}{k}$ with probability $1/2$ and uniformly chosen from $\binom{[2k-1]}{k-1}$ with probability $1/2$.
\end{lemma}
\begin{proof}
Suppose for the sake of contradiction that we have an algorithm $\cA$ that solves $\Threshold{2k-1}{k}$ with error probability $\delta = o(1)$ and uses only $(1-\epsilon)\frac{2k \log \frac{k}{\delta}}{\DKL}$ noisy queries in expectation, for some absolute constant $\epsilon>0$. Let $\cD$ be a distribution of inputs where with $1/2$ probability the input is chosen uniformly from $\binom{[2k-1]}{k}$, and with $1/2$ probability the input is chosen uniformly from $\binom{[2k-1]}{k-1}$. Now we consider two cases, depending on whether the expected number of queries $\cA$ make on indices with $1$'s is larger or not. In either case, we will use $\cA$ to obtain an algorithm more efficient than the lower bound in \cref{thm:th:small-k}, thus reaching a contradiction.

First, suppose $\cA$ makes more queries in expectation on indices with $1$'s under input distribution $\cD$. Let $k' = \Theta(k / \log k)$ and let $n = k - 1 + k'$. Consider an instance of $\Threshold{n}{k'}$ where the input is uniformly chosen from $\binom{[n]}{k'}$ with probability $1/2$ and uniformly chosen from $\binom{[n]}{k'-1}$ with probability $1/2$. By \cref{thm:th:small-k}, this instance requires $(1-o(1)) \frac{n \log \frac{k'}{\delta}}{\DKL}=(1-o(1)) \frac{k \log \frac{k}{\delta}}{\DKL}$ queries in expectation. We will design an algorithm $\cB$ solving such an instance utilizing $\cA$.

When $\cB$ gets the input, it first adds $2k-1-n$ $1$'s to the input, and then randomly shuffle the indices. Then $\cB$ sends this input to $\cA$. Whenever $\cA$ makes a query to a $1$ that is artificially added, $\cB$ simulates a noisy query using random bits; when $\cA$ makes a query to an actual input, $\cB$ makes a query as well and pass the result to $\cA$. When $\cA$ returns a result, $\cB$ returns the same result. It is not difficult to verify that the input distribution for $\cA$ is exactly $\cD$. Also, whenever the input to $\cA$ has at least $k$ $1$'s, the input to $\cB$ has $k'$ $1$'s, and vice versa, so the correct output of $\cA$ is the same as  the correct output of $\cB$. Therefore, $\cB$ is correct whenever $\cA$ is correct, which happens with probability $\delta$.

Next, we analyze the expected number of queries $\cB$ makes, which consist of two parts:
\begin{itemize}
  \item The number of queries $\cA$ makes to an actual $0$ in the input: Because we are in the case where $\cA$ makes more queries in expectation on indices with $1$'s in the input than indices with $0$'s, the expected number of this type of queries is at most half of the expected total number of queries $\cA$ makes. Thus, the number of queries in this case is at most $(1-\epsilon)\frac{k \log \frac{k}{\delta}}{\DKL}$.
  \item The number of queries $\cA$ makes to an actual $1$ in the input: Each actual $1$ in the input to $\cB$ is later permuted to a random position in the sequence. By symmetry, a random $1$ in the input to $\cA$ under the input distribution $\cD$ is queried at most $\frac{1}{k-1} \cdot (1-\epsilon)\frac{2k \log \frac{k}{\delta}}{\DKL}$ times in expectation. The number of actual $1$'s is $k' = \Theta(k / \log k)$, so the expected number of queries $\cA$ makes on them is $\frac{k'}{k-1} \cdot (1-\epsilon)\frac{2k \log \frac{k}{\delta}}{\DKL} = o(1) \cdot \frac{k \log \frac{k}{\delta}}{\DKL}$.
\end{itemize}
As a result, the total number of queries $\cB$ makes is $(1-\epsilon)\frac{k \log \frac{k}{\delta}}{\DKL}$, which contradicts \cref{thm:th:small-k}.

For the second case where $\cA$ makes more (or equal number of) queries in expectation on indices with $0$'s under input distribution $\cD$. The only difference is that, when $\cB$ sends the input to $\cA$, it has to flip the roles of $0$'s and $1$'s. Additionally, it has to flip the result $\cA$ returns. We omit the details.
\end{proof}

Given \cref{thm:th:small-k} and \cref{lem:th:maj}, we are ready to prove the lower bound for general $k$.
\begin{theorem}
\label{thm:th:large-k}
  Solving $\Threshold{n}{k}$ with $
\delta = o(1)$ error probability requires
\[
(1-o(1))\frac{n \log \frac{k}{\delta}}{\DKL}
\]
noisy queries in expectation for $n/\log n \le k \le n / 2$.
\end{theorem}
\begin{proof}
  The high-level proof strategy is similar to that of \cref{lem:th:maj}, by reducing from a hard instance to $\Threshold{n}{k}$. However, the difference is that we need to reduce from two different cases depending on whether the average number of queries per $1$ or per $0$ is larger. In \cref{lem:th:maj} we did not have to do it because we can simply flip all the $0$'s and $1$'s in the input and retain the same problem as $n = 2k - 1$.

  Let $\cD$ be the input distribution where with $1/2$ probability the input is chosen uniformly from $\binom{[n]}{k}$ and with $1/2$ probability the input is chosen uniformly from $\binom{[n]}{k - 1}$. Let $\cA$ be an algorithm solving $\Threshold{n}{k}$ under input distribution $\cD$ using $(1-\epsilon)\frac{n \log \frac{k}{\delta}}{\DKL}$ noisy queries in expectation, for some absolute constant $\epsilon>0$. Let $q_0$ denote the expected number of queries $\cA$ makes on a random index with input value $0$ (under input distribution $\cD$), and let $q_1$ denote the expected number of queries $\cA$ makes on a random index with input value $1$ (under input distribution $\cD$). Let $Q$ be the expected number of queries $\cA$ makes.

  Consider the following two cases: $q_0\le q_1$ and $q_0>q_1$.

  \paragraph{Case $q_0 \le q_1$.}

  Let $k' = \Theta(n / \log n) \le k$ and let $n' = n - k + k'$. Note that $n' \in [n / 2, n]$, so we have $k' = \Theta(n' / \log n')$. By \cref{thm:th:small-k}, solving $\Threshold{n'}{k'}$ under input distribution where with $1/2$ probability the input is uniformly from $\binom{[n']}{k'}$ and with $1/2$ probability the input is uniformly from $\binom{[n']}{k' - 1}$ with error probability $\delta$ requires $(1-o(1))\frac{n' \log \frac{k'}{\delta}}{\DKL} = (1-o(1))\frac{n' \log \frac{n}{\delta}}{\DKL}$ noisy queries in expectation.

  Given an instance of $\Threshold{n'}{k'}$, we add $n - n'$ artificial $1$'s to the input, and then randomly permute the input, and feed it to $\cA$. If $\cA$ queries an actual input, we also make an actual query; if $\cA$ queries an artificial input, we can simulate a query without making an actual query. Finally, we use the result returned by $\cA$ as our answer. It is not difficult to verify that this algorithm is correct with error probability $\delta$, and the input distribution to $\cA$ is $\cD$.

  Let us analyze the expected number of queries $Q'$ used by the algorithm, which can be expressed as follows:
  \[
  \frac{1}{2} \cdot \left(q_1 (k' - 1) + q_0 (n - k + 1) \right) + \frac{1}{2} \cdot \left(q_1 k' + q_0 (n - k) \right).
  \]
  Also, notice that the number of queries $Q$ made by $\cA$ under input $\cD$ is
  \[
  \frac{1}{2} \cdot \left(q_1 (k - 1) + q_0 (n - k + 1) \right) + \frac{1}{2} \cdot \left(q_1 k + q_0 (n - k) \right).
  \]
  As $q_0 \le q_1$, the above implies that $q_1 \ge Q / n$. Furthermore, we have that $Q - Q' = q_1 \cdot (k - k') = q_1 \cdot (n - n') \ge \frac{n - n'}{n} \cdot Q$. Therefore,
  \[
  Q' \le \frac{n'}{n} \cdot Q \le (1-\epsilon)\frac{n' \log \frac{k}{\delta}}{\DKL} \le  (1-\epsilon)\frac{n' \log \frac{n}{\delta}}{\DKL},
  \]
  which contradicts \cref{thm:th:small-k}.

  \paragraph{Case $q_0 > q_1$.}
  By \cref{lem:th:maj}, solving $\Threshold{2k-1}{k}$ under input distribution where with $1/2$ probability the input is uniformly from $\binom{[2k-1]}{k}$ and with $1/2$ probability the input is uniformly from $\binom{[2k-1]}{k - 1}$ with error probability $\delta$ requires $(1-o(1))\frac{2k \log \frac{k}{\delta}}{\DKL} = (1-o(1))\frac{2k \log \frac{n}{\delta}}{\DKL}$ noisy queries in expectation.

  Given such an input to $\Threshold{2k-1}{k}$, we add $n - 2k+1$ artificial $0$'s to the input, and then randomly permute the input, and feed it to $\cA$. Similar as before, we make an actual query if $\cA$ queries an actual input, and we simulate a query otherwise. It is not difficult to verify that this algorithm is correct with error probability $\delta$, and the input distribution to $\cA$ is $\cD$.

  The expected number of actual queries used by the algorithm can be similarly analyzed as the previous case, which can be upper bounded by
  $(1-\epsilon)\frac{2k \log \frac{n}{\delta}}{\DKL},
  $
  contradicting \cref{lem:th:maj}.
\end{proof}

\subsection{Upper bound for \texorpdfstring{$k$}{k}-Threshold}
In this section we prove our upper bound for $k$-Threshold, stated as follows.

\begin{theorem}
\label{thm:threshold-upper}
  Given a sequence $a \in \{0, 1\}^n$ and an integer $1 \le k \le n$, there is an algorithm that can output $\min\{k, \lVert a \rVert_1\}$ with error probability $\delta = o(1)$ using
  \[
  (1+o(1))\frac{n \log \frac{k}{\delta}}{\DKL}
  \]
  noisy queries in expectation.
\end{theorem}

\subsubsection{Preliminaries}
The following lemmas are standard.

\begin{lemma}[e.g., \cite{gu2023optimal, wang2024noisy}]
\label{lem:check-bit}
    For a bit $B$, there is an algorithm $\textup{\textsc{Check-Bit}}(B, \delta)$ that can return the value of the bit with error probability $\le \delta$ using
    \[
    (1+o_{1/\delta}(1)) \frac{\log \frac{1}{\delta}}{\DKL}
    \]
    noisy queries in expectation.
\end{lemma}

\begin{lemma}[e.g., \cite{feller}]
\label{lem:monkey-at-cliff}
Consider a biased random walk on $\bZ$ starting at $0$. At each time step, the walk adds $1$ to the current value with probability $p < 1/2$, and adds $-1$ to the current value with probability $1-p$. Then the probability that the random walk ever reaches some integer $x \ge 0$ is $(p/(1-p))^x$.
\end{lemma}

\begin{lemma}[e.g., \cite{feller}]
\label{lem:expected-hitting-time}
Consider a biased random walk on $\bZ$ starting at $0$. At each time step, the walk adds $1$ to the current value with probability $p < 1/2$, and adds $-1$ to the current value with probability $1-p$. Then the expected number of steps needed to first reach some integer $-x \le 0$ is $\frac{x}{1-2p}$.
\end{lemma}

Our algorithm for $k$-Threshold and Counting uses the following asymmetric version of $\textsc{Check-Bit}$.
\begin{lemma}
\label{lem:asymmetric-check-bit}
    For a bit $B$, there is an algorithm $\textup{\textsc{Asymmetric-Check-Bit}}(B, \delta_0, \delta_1)$ such that:
    \begin{itemize}
        \item If the actual value of $B$ is $0$, then the algorithm returns the value of the bit with error probability $\le \delta_0$
        using
        \[
        (1+o_{1/\delta_1}(1)) \frac{\log \frac{1}{\delta_1}}{\DKL}
        \]
        noisy queries in expectation.
        \item If the actual value of $B$ is $1$, then the algorithm returns the value of the bit with error probability $\le \delta_1$
        using
        \[
        (1+o_{1/\delta_0}(1)) \frac{\log \frac{1}{\delta_0}}{\DKL}
        \]
        noisy queries in expectation.
    \end{itemize}
\end{lemma}
\begin{proof}
Let $a$ and $b$ be two integer parameters to be set later. The algorithm works as follows: we keep querying the input bit, and keep track of the number of queries that returns $1$ (denoted by $q_1$) and the number of queries that returns $0$ (denoted by $q_0$). We stop once $q_1 - q_0 = -a$, in which case we declare the bit to be $0$, or $q_1 - q_0 = b$, in which case we declare the bit to be $1$.

If the input bit is $1$, then the probability that the algorithm returns $0$ can be upper bounded by
$
\left(\frac{p}{1-p}\right)^{a}
$
using \cref{lem:monkey-at-cliff},
so by setting $a = \left\lceil \frac{\log(1/\delta_1)}{\log((1-p) / p)} \right\rceil$, this probability is upper bounded by $\delta_1$.

On the other hand, if the input bit is $0$, then the probability that the algorithm returns $1$ can be upper bounded by
$
\left(\frac{p}{1-p}\right)^{b}
$
using \cref{lem:monkey-at-cliff},
so by setting $b = \left\lceil \frac{\log(1/\delta_0)}{\log((1-p) / p)} \right\rceil $, this probability is upper bounded by $\delta_0$.

Now we consider the expected running time of the algorithm. First suppose the input bit is $1$, then the algorithm can be viewed as a random walk on integers starting at $0$, and each time it adds $1$ with probability $1-p$ and subtracts $1$ with probability $p$. It stops once the random walk reaches $-a$ or $b$. This stopping time is upper bounded by the first time it reaches $b$, and the expected number of steps required for it to first reach $b$ is $\frac{b}{1-2p} = (1+o_{1/\delta_0}(1)) \frac{\log(1/\delta_0)}{\DKL}$ by \cref{lem:expected-hitting-time}. Similarly, if the input bit is $0$, then the expected number of queries used by the algorithm is upper bounded by $(1+o_{1/\delta_1}(1)) \frac{\log(1/\delta_1)}{\DKL}$.
\end{proof}

\subsubsection{The proof}
We are now ready to state our proof of \cref{thm:threshold-upper}.

Our algorithm is extremely simple and is outlined in \cref{algo:threshold}. It repeatedly calls the subroutine $\textsc{Asymmetric-Check-Bit}(a_i, \delta / 2n, \delta / 2k)$ from \cref{lem:asymmetric-check-bit} for every $i \in [n]$, and uses the returned value as the guess for $a_i$. If at any point, the current number of $a_i$ whose guess is $1$ reaches $k$, the algorithm terminates early and return $k$. Otherwise, the algorithm returns the number of $a_i$ whose guess is $1$ at the end.
\begin{algorithm}[ht]
\caption{}
\label{algo:threshold}
\begin{algorithmic}[1]
\Procedure{Threshold-Count}{$\{a_1, \ldots, a_n\}, k, \delta$}
\State $cnt = 0$
\For{$i = 1 \to n$}
\State $cnt = cnt + \textsc{Asymmetric-Check-Bit}(a_i, \delta / 2n, \delta / 2k)$
\If{$cnt \ge k$}
\Return{$k$}
\EndIf
\EndFor
\State \Return{$cnt$}
\EndProcedure
\end{algorithmic}
\end{algorithm}
\paragraph{Error probability.} We first analyze the error probability of the algorithm. First, suppose $\lVert a \rVert_1 \ge k$. Let $S \subseteq [n]$ be an arbitrary size-$k$ set where $a_i = 1$ for $i \in S$. Then by the guarantee of \textsc{Asymmetric-Check-Bit}, for every $i \in S$, the probability that the guess for $a_i$ is not $1$ is $\le \delta / 2k$. Therefore, by union bound, the guesses for $a_i$ for all $i \in S$ are $1$ with probability $\ge 1-\delta/2$. Therefore, the count will be at least $k$ so the algorithm will return $k$ as the correct answer with error probability $\le \delta / 2 \le \delta$.

If $\lVert a \rVert_1 < k$, then by union bound, the probability that all guesses are correct is $\ge 1 - \frac{\delta}{2k} \cdot \lVert a \rVert_1 - \frac{\delta}{2n} \cdot (1 - \lVert a \rVert_1) \ge 1-\delta$, so the returned count of the algorithm is also correct with probability $\ge 1- \delta$.

\paragraph{Expected number of queries.} We first consider the number of $a_i$ with $a_i = 1$ that we pass to \textsc{Asymmetric-Check-Bit}. One trivial upper bound is $\lVert a \rVert_1$. Also, the expected number of $a_i = 1$ we need to pass to \textsc{Asymmetric-Check-Bit} before $cnt$ is incremented by $1$ is $\le \frac{1}{1-\delta/2k}$, so another upper bound is $\frac{k}{1-\delta/2k} \le (1+o(1)) k$ (as $\delta = o(1)$). On the other hand, the number of $a_i$ where $a_i = 0$ that we pass to \textsc{Asymmetric-Check-Bit} is upper bounded by $n - \lVert a \rVert_1$. Note that whether we pass some $a_i$ to \textsc{Asymmetric-Check-Bit} only depends on the queries we make to $a_{i'}$ for $i' < i$, so even given that we pass some $a_i$ to \textsc{Asymmetric-Check-Bit}, we can still use the bounds from \cref{lem:asymmetric-check-bit} to bound the expected number of queries we make to $a_i$. Therefore, the expected number of queries can be upper bounded by
\[
(1+o(1)) \left(\min\left\{\lVert a \rVert_1, k\right\} \cdot  \frac{\log\frac{n}{\delta}}{\DKL} + (n - \lVert a \rVert_1) \cdot \frac{\log \frac{k}{\delta}}{\DKL}\right).
\]
Then we consider two cases depending on how large $k$ is.

\paragraph{Case $k \ge n/\log n$.}
  In this case, $\log(n) = (1+o(1)) \log k$, so the expected number of queries can be upper bounded by
  \begin{align*}
& (1+o(1)) \left(\min\left\{\lVert a \rVert_1, k\right\} \cdot  \frac{\log \frac{k}{\delta}}{\DKL} + (n - \lVert a \rVert_1) \cdot \frac{\log \frac{k}{\delta}}{\DKL}\right) \\
= & (1+o(1)) \frac{n \log \frac{k}{\delta}}{\DKL}.
\end{align*}

\paragraph{Case $k \le n/\log n$.}
In this case, $k \log \frac{n}{\delta} = o(n \log \frac{k}{\delta})$, so the expected number of queries can be upper bounded by
\begin{align*}
& (1+o(1)) \left(k \cdot  \frac{\log\frac{n}{\delta}}{\DKL} + n \cdot \frac{\log \frac{k}{\delta}}{\DKL}\right) \\
= & (1+o(1)) \frac{n \log \frac{k}{\delta}}{\DKL}.
\end{align*}

\subsection{Bounds for Counting}

Let us first prove a one-sided upper bound for Counting.
\begin{theorem}
\label{thm:counting-upper-1}
  Given a sequence $a \in \{0, 1\}^n$, there is an algorithm that can output $\lVert a \rVert_1$ with error probability $\delta = o(1)$ using
  \[
  (1+o(1))\frac{n \log \frac{\lVert a \rVert_1 + 1}{\delta}}{\DKL}
  \]
  noisy queries in expectation.
\end{theorem}
The difference between \cref{thm:counting-upper-1} and the upper bound part of \cref{thm:counting} is that we have $\lVert a \rVert_1 + 1$ rather than $\min\{\lVert a \rVert_1 + 1, n-\lVert a \rVert_1 + 1\}$ inside the log term.
\begin{proof}

Outlined in \cref{algo:counting}, the algorithm for counting is an adaptation of \cref{algo:threshold}. Notice that in \cref{algo:threshold}, when calling $\textsc{Asymmetric-Check-Bit}$, we have $\delta_0 = \delta/2n$ and $\delta_1 = \delta / 2k$. The value $\delta_0$ is fixed regardless of the value of $k$, and the value $\delta_1$ depends on $k$. The algorithm for counting in some sense is simulating \cref{algo:threshold}, but dynamically adjusting $\delta_1$ based on the current estimate of $k$, which is the number of input bits that are believed to be $1$'s.

For simplicity, throughout the analysis, we use $k^* = \lVert a \rVert_1$ to denote the desired answer. Let $S_0 \subseteq [n]$ be indices $i$ with $a_i = 0$ and let $S_1 = [n] \setminus S_0$.

\begin{algorithm}[ht]
\caption{}
\label{algo:counting}
\begin{algorithmic}[1]
\Procedure{Counting}{$\{a_1, \ldots, a_n\}, \delta$}
\State $k \gets 0$
\State $c \gets \{0\}^n$
\State $Active \gets [n]$
\While{True}
\State $i^* \gets \argmax_{i \in Active} c_i$
\If{$c_{i^*} \le -\frac{\log (6(k+1) / \delta)}{\log((1-p)/p)}$}
\Return{$k$}
\EndIf
\If{$\textsc{Query}(a_{i^*}) = 1$}
\State $c_{i^*} = c_{i^*} + 1$
\If{$c_{i^*} \ge \frac{\log (6 n / \delta)}{\log((1-p)/p)}$}
\State $Active \gets Active \setminus \{i^*\}$
\State $k \gets k + 1$
\EndIf
\Else
\State $c_{i^*} = c_{i^*} - 1$
\EndIf
\EndWhile
\EndProcedure
\end{algorithmic}
\end{algorithm}

\paragraph{Error probability.} For any $i$, we can view the value of $c_i$ as a random walk, i.e., after every query to $a_i$, we either adds $1$ to $c_i$ or subtracts $1$ from $c_i$. If $c_i = 1$, then the probability that we add $1$ to $c_i$ after every step is $1-p$, and the probability that we subtract $1$ is $p$. If $c_i = 0$, then the probability that we add $1$ to $c_i$ after every step is $p$, and the probability that we subtract $1$ is $1-p$. Conceptually, we can view $c_i$ as an infinite random walk, and the algorithm only utilizes some prefix of it.

For $i \in S_1$, by \cref{lem:monkey-at-cliff}, the probability that $c_i$ ever reaches $-\left\lceil \frac{\log (6 / \delta)}{\log((1-p)/p)} \right\rceil$ is at most
\[
\left( \frac{p}{1-p}\right)^{\frac{\log (6 / \delta)}{\log((1-p)/p)} } = \delta / 6.
\]
Furthermore, if some $c_i$ never reaches $-\left\lceil \frac{\log (6 / \delta)}{\log((1-p)/p)} \right\rceil$, then $i$ will eventually be removed from $Active$ and contributes towards $k$. Since the random walks $c_i$'s are independent for different $i$'s, the probability that there are $\lceil k^* /2 \rceil + 1$ many $i \in S_1$ where  $c_i$ reaches $-\left\lceil \frac{\log (6 / \delta)}{\log((1-p)/p)} \right\rceil$ can be upper bounded by
\[
\binom{k^*}{\lceil k^* /2 \rceil + 1} \cdot (\delta / 6)^{\lceil k^* /2 \rceil + 1} \le 2^{k^*} (\delta / 6)^{ k^* /2  + 1} \le (1.5\delta)^{k^* /2} \cdot (\delta / 6) \le \delta / 6,
\]
where the last step holds because $\delta = o(1)$ and hence we can assume $1.5 \delta \le 1$.
Thus, up to $\delta / 6$ error probability, the number of $i \in S_1$ that contributes towards $k$ is at least $k^* - (\lceil k^* /2 \rceil + 1)$, which means the final value of $k$ is $\ge \max\{0, k^* - (\lceil k^* /2 \rceil + 1)\} = \max\{0, \lfloor k^* / 2 \rfloor - 1\}$. This further implies that $k + 1 \ge k^* / 3$. Let $E_1$ be the event $k + 1 \ge k^* / 3$ and as analyzed above, $\bP(\neg E_1) \le \delta / 6$.

Next, we consider the event $E_2$ where all $c_i$ for $i \in S_1$ never reaches $-\left\lceil \frac{\log (6 (k^* / 3) / \delta)}{\log((1-p)/p)} \right\rceil$. By \cref{lem:monkey-at-cliff}, the probability that each $c_i$ for $i \in S_1$ reaches $-\left\lceil \frac{\log (6 (k^* / 3) / \delta)}{\log((1-p)/p)} \right\rceil$ is upper bounded by
\[
\left( \frac{p}{1-p}\right)^{\frac{\log (6 (k^* / 3) / \delta)}{\log((1-p)/p)}} = \frac{\delta}{2 k^*}.
\]
Therefore, by union bound, $\bP(\neg E_2) \le \delta /2$.

Next, let $E_3$ be the event where all $c_i$ for $i \in S_0$ never reaches $\left\lceil \frac{\log (6 n / \delta)}{\log((1-p)/p)} \right\rceil$. By \cref{lem:monkey-at-cliff}, the probability that each $c_i$ for $i \in S_0$ reaches $\left\lceil \frac{\log (6 n / \delta)}{\log((1-p)/p)} \right\rceil$ is upper bounded by
\[
\left( \frac{p}{1-p}\right)^{\frac{\log (6 n / \delta)}{\log((1-p)/p)}} = \frac{\delta}{6n}.
\]
Therefore, by union bound, $\bP(\neg E_3) \le \delta /6$.

Assume $E_1, E_2, E_3$ all happen. Then we know that $k + 1 \ge k^* / 3$, and so all $i \in S_1$ will eventually be removed from $Active$ and contribute towards $k$. Also, all $i \in S_0$ will not contribute towards $k$. Therefore, the returned value of $k$ will be equal to $k^*$. Hence, the error probability of the algorithm is upper bounded by $\bP(\neg E_1 \vee \neg E_2 \vee \neg E_3) \le \delta / 6 + \delta / 2 + \delta / 6 \le \delta$.

\paragraph{Expected number of queries.} Next, we analyze the expected number of queries used by the algorithm.

First, for every $i \in S_1$, the expected number of times we query $a_i$ is upper bounded by the expected number of steps it takes for the random walk $c_i$ takes to reach $\left\lceil \frac{\log (6 n / \delta)}{\log((1-p)/p)} \right\rceil$. By \cref{lem:expected-hitting-time}, it is bounded by
\[
\frac{\left\lceil \frac{\log (6 n / \delta)}{\log((1-p)/p)} \right\rceil}{1 - 2p} = (1+o(1)) \frac{\log \frac{n}{\delta}}{\DKL}.
\]

Then we consider the expected number of times we query $a_i$ for $i \in S_0$. Fix any $i \in S_0$. Let $F_j$ be the event where for exactly $j$ distinct $i' \in S_0 \setminus \{i\}$, the infinite random walk $c_j$ ever reaches $\left\lceil \frac{\log (6 n / \delta)}{\log((1-p)/p)} \right\rceil$. By analysis in the error probability part, $\bP(F_j) \le \binom{n}{j} \cdot (\delta / 6 n)^j \le \delta^j$. Note that $F_j$ is independent with the random walk $c_i$. If $F_j$ holds, then the number of times we query $a_i$ is upper bounded by the expected number of steps it takes for the random walk $c_i$ first hits $\left\lceil \frac{\log (6 (k^* + j + 1)) / \delta)}{\log((1-p)/p)} \right\rceil$ (because once $c_i$ hits this value, $k$ can never be larger than $k^* + j$ under $F_j$, so we will not query $a_i$ again). By \cref{lem:expected-hitting-time}, this expectation is
\[
\frac{\left\lceil \frac{\log (6 (k^* + j + 1)) / \delta)}{\log((1-p)/p)} \right\rceil}{1 - 2p} = (1+o(1)) \frac{\log \frac{k^* + j + 1}{\delta}}{\DKL}.
\]
Let $Q_i$ be the number of times we query $a_i$. Then we have
\begin{align*}
  \bE(Q_i) &= \sum_{j \ge 0} \bE(Q_i | F_j) \cdot \bP(F_j)\\
  &\le \frac{1+o(1)}{\DKL} \sum_{j \ge 0}  \delta^j \cdot \log \left(\frac{k^* + j + 1}{\delta} \right)\\
  &\le \frac{1+o(1)}{\DKL} \sum_{j \ge 0}  \delta^j \cdot \log \left(\frac{(k^* + 1) (j + 1)}{\delta} \right)\\
  &= \frac{1+o(1)}{\DKL} \left( \sum_{j \ge 0}  \delta^j \cdot \log \left(\frac{k^* + 1}{\delta} \right) + \sum_{j \ge 0} \delta^j \cdot \log(j+1)\right) \\
  &\le \frac{1+o(1)}{\DKL} \left( \frac{1}{1-\delta} \cdot \log \left(\frac{k^* + 1}{\delta} \right) + O(1)\right) \\
  &\le (1+o(1)) \frac{\log \frac{k^*+1}{\delta}}{\DKL}.
\end{align*}

Summing up everything, the overall expected number of queries is
\[(1+o(1)) \left(k^* \cdot  \frac{\log \frac{n}{\delta}}{\DKL} + (n - k^*) \cdot \frac{\log\frac{k^*+1}{\delta}}{\DKL}\right).
\]
Similar to the proof of \cref{thm:threshold-upper}, we consider two cases depending on how large $k^*$ is.

\paragraph{Case $k^* \ge n/\log n$.}
In this case, $\log(n) = (1+o(1)) \log (k^*+1)$, so the expected number of queries can be bounded by
  \begin{align*}
    & (1+o(1)) \left(k^* \cdot  \frac{\log \frac{k^*+1}{\delta}}{\DKL} + (n - k^*) \cdot \frac{\log \frac{k^*+1}{\delta}}{\DKL}\right)\\
    =& (1+o(1)) \frac{n \log \frac{k^*+1}{\delta}}{\DKL}.
  \end{align*}

\paragraph{Case $k^* \le n/\log n$.}
In this case, $k^* \log \frac{n}{\delta} = o(n \log \frac{k^*+1}{\delta})$, so the expected number of queries can be bounded by
   \begin{align*}
    & (1+o(1)) \left(k^* \cdot  \frac{\log \frac{n}{\delta}}{\DKL} + n \cdot \frac{\log \frac{k^*+1}{\delta}}{\DKL}\right)\\
    =& (1+o(1)) \frac{n \log \frac{k^*+1}{\delta}}{\DKL}.
  \end{align*}
\end{proof}

Finally, we prove \cref{thm:counting}, which we recall below:

\Counting*
\begin{proof}

Suppose there is an algorithm for Counting with
\[
(1-\epsilon)\frac{n \log \frac{\min\{\lVert a \rVert_1, n - \lVert a \rVert_1\} + 1}{\delta}}{\DKL}
\]
noisy queries in expectation, for some absolute constant $c>0$. First, suppose $\lVert a \rVert_1 \le n / 2$. Recall that in the lower bound for $k$-Threshold, the hard distribution for $k$-Threshold is to distinguish whether the input contains $k$ or $k-1$ $1$'s. Thus, by running the assumed algorithm for Counting, the running time would be
\[
(1-\epsilon)\frac{n \log \frac{k + 1}{\delta}}{\DKL},
\]
which contradicts the lower bound for $k$-Threshold.

If $\lVert a \rVert_1 > n / 2$, we can simply first flip all the bits in the $k$-Threshold instance, and then use the same argument.

For the upper bound, given \cref{thm:counting-upper-1}, it suffices to first estimate whether $\lVert a \rVert_1$ is bigger than $n/2$ or smaller. When it is smaller, we can directly run \cref{thm:counting-upper-1}; otherwise, we first flip all input bits and then run \cref{thm:counting-upper-1}.

More precisely, we randomly sample $n^{0.99}$ input elements with replacements, use \cref{lem:check-bit} to estimate them with error probability $1/n^{100}$. If the fraction of elements in the sample whose estimates are $1$ is $\le \frac{1}{2}$, we directly call \cref{thm:counting-upper-1} with error bound $\delta$ and return the result; otherwise, we flip all input bits, call \cref{thm:counting-upper-1} with error bound $\delta$, and return $n$ minus the result.

Regardless of whether we flip the input bits, the output is always correct assuming the returned result of \cref{thm:counting-upper-1} is correct. Therefore, the error probability of the algorithm is at most $\delta$.

Next, we analyze the expected running time of the algorithm. Let $E_1$ be the event that the fraction of sampled elements that are $1$ is within $n^{-0.01}$ of the fraction of all elements that are $1$. By Chernoff bound,
\[
\bP(\neg E_1) \le 2 e^{-2 \left(n^{-0.01} \right)^2 \cdot n^{0.99}} \le O\left(\frac{1}{n^{99}}\right).
\]
Let $E_2$ be the event that the returned result of \cref{lem:check-bit} for all sampled elements are correct. By union bound,
\[
\bP(\neg E_2) \le \frac{1}{n^{99}}.
\]
If $E_1$ and $E_2$ both hold, the expected number of queries of \cref{thm:counting-upper-1} can be bounded as
\[
(1+o(1)) \cdot \frac{n}{\DKL} \cdot
\begin{cases}
   \log \frac{\lVert a \rVert_1 + 1}{\delta},  & \text{if } \lVert a\rVert_1 \le n / 2 - n^{0.99},\\
   \log \frac{n + 1}{\delta},  & \text{if } n / 2 - n^{0.99} < \lVert a\rVert_1 \le n / 2 + n^{0.99},\\
   \log \frac{n - \lVert a \rVert_1 + 1}{\delta},  & \text{if } \lVert a\rVert_1 > n / 2 + n^{0.99}.
\end{cases}
\]
Regardless of which case it is, the expected number of queries is always $(1+o(1))\frac{n \log \frac{\min\{\lVert a \rVert_1, n - \lVert a \rVert_1\} + 1}{\delta}}{\DKL}$. Even if $E_1$ and $E_2$ do not both hold, the expected number of queries of \cref{thm:counting-upper-1} can be bounded as
\[
(1+o(1)) \cdot \frac{n \log \frac{n + 1}{\delta}}{\DKL}.
\]

Therefore, the overall expected number of queries can be bounded as
\begin{align*}
& n^{0.99} \cdot (1+o(1)) \cdot \frac{\log \frac{1}{n^{100}}}{\DKL} + (1+o(1))\frac{n \log \frac{\min\{\lVert a \rVert_1, n - \lVert a \rVert_1\} + 1}{\delta}}{\DKL} \\
& + \bP(\neg E_1 \vee \neg E_2) (1+o(1))\frac{n \log \frac{n + 1}{\delta}}{\DKL}\\
= & (1+o(1))\frac{n \log \frac{\min\{\lVert a \rVert_1, n - \lVert a \rVert_1\} + 1}{\delta}}{\DKL}
\end{align*}
as desired.
\end{proof}

\ifdefined\isarxiv
\section*{Acknowledgments}
We thank Ziao Wang for pointing out an error in the statement of \cref{open:k-selection} in a previous version of the paper.
\else
\fi

\bibliographystyle{alpha}
\bibliography{ref}

\appendix
\section{A simpler proof of hardness of Threshold for small \texorpdfstring{$k$}{k}}
\label{sec:th-small}
In this section we provide a simpler proof of \cref{thm:th:small-k}. Let us recall the theorem statement.

\ThresholdSmallK*

Our proof uses the three-phase framework which we also used in the proof of \cref{thm:graph-conn-hard}. It is a refinement of \cite{feige1994computing}'s two-phase framework, which were used to prove that computing $\Threshold{n}{k}$ requires $\Omega\left(n \log \frac k\delta\right)$ noisy queries.
Here we add a phase where the oracle can send extra information to the algorithm, allowing for a more precise analysis obtaining the exact constant.

\paragraph{Comparison with \cite{wang2024noisy}'s proof.}
Let us briefly compare our proof with \cite{wang2024noisy}'s proof. Both proofs are based on \cite{feige1994computing}'s two-phase framework.
\cite{wang2024noisy}'s proof is divided into three cases: $\log(1/\delta)\log\log(1/\delta) < \log k \le \frac{\log(1/\delta)}{\log\log(1/\delta)}$, $\log k > \log(1/\delta)\log\log(1/\delta)$, $\log k \le \frac{\log(1/\delta)}{\log\log(1/\delta)}$.
The first two cases are handled using the two-phase framework, and the last case is proved using Le Cam's two point method.
In comparison, our proof is simpler and handles all $k=o(n)$ and $\delta = o(1)$ uniformly. We achieve this simplification by carefully designing the information that the oracle reveals to the algorithm for free.

\subsection{Three-phase problem} \label{sec:th-small:three-phase}
Let us now describe the three-phase problem. Let $\epsilon_1,\epsilon_2>0$ be two absolute constants (i.e., they do not grow with $n$).
\begin{enumerate}
  \item In Phase 1, the algorithm makes $m_1 = (1-\epsilon_1) \frac{\log \frac k\delta}{\DKL}$ queries to every element.
  \item In Phase 2, the oracle reveals some elements to the algorithm.
  \item In Phase 3, the algorithm makes $m_2 = (1-\epsilon_2) n$ adaptive exact queries.
\end{enumerate}
The goal of the algorithm is to distinguish whether there are at least $k+1$ ones among the elements.
\begin{lemma} \label{lem:th:reduction}
  If no algorithm can solve the three-phase problem with error probability $\delta>0$, then no algorithm can solve $\Threshold{n}{k}$ with error probability $\delta$ using at most $(1-\epsilon_1)(1-\epsilon_2) \frac{n \log \frac k\delta}{\DKL}$ noisy queries.
\end{lemma}
The proof is similar to \cref{lem:inf:reduction} and omitted.

By \cref{lem:th:reduction}, to prove \cref{thm:th:small-k}, it suffices to prove hardness of the three-phase problem.
\begin{proposition} \label{prop:th:three-phase-hard}
  For any absolute constants $\epsilon_1,\epsilon_2>0$, no algorithm can solve the three-phase problem for $\Threshold{n}{k}$ with error probability $\delta$ for $k=o(n)$ and $\delta=o(1)$, where the input is uniformly chosen from $\binom{[n]}k$ with probability $1/2$ and uniformly chosen from $\binom{[n]}{k-1}$ with probability $1/2$.
\end{proposition}
\cref{thm:th:small-k} follows by combining \cref{lem:th:reduction} and \cref{prop:th:three-phase-hard}.

The rest of the section is devoted to the proof of \cref{prop:th:three-phase-hard}.

\subsection{Phase 1} \label{sec:th-small:phase-1}
Define $A = \{i\in [n]: a_i = 1\}$ where $a$ is the input bit string.
Then $A \in \binom{[n]}{k} \cup \binom{[n]}{k-1}$, and the goal is to distinguish whether $|A|=k$ or $k-1$.

In Phase 1, the algorithm makes $m_1 = (1-\epsilon_1) \frac{\log \frac k\delta}{\DKL}$ queries to every element $i\in [n]$.
Let $a_i$ denote the number of times where a query to $i$ returns $1$.
Then for $i\in A$, $a_i \sim \Bin(m_1, 1-p)$; for $i\not \in A$, $a_i\sim \Bin(m_1, p)$.
For $0\le j\le m_1$, define
\begin{align*}
  p_j = \bP(\Bin(m_1,1-p)=j) = \binom{m_1}j (1-p)^j p^{m_1-j}.
\end{align*}
Let $I = \left[p m_1 - m_1^{0.6}, p m_1 + m_1^{0.6}\right]$.
\begin{lemma} \label{lem:th:binom-concentrate}
  Let $x\sim \Bin(m_1, 1-p)$, $y\sim \Bin(m_1, p)$.
  Then
  \begin{align}
    \label{eqn:lem-th-binom-concentrate:i} \bP(x\in I) &= (\delta/k)^{1-\epsilon_1 \pm o(1)}, \\
    \label{eqn:lem-th-binom-concentrate:ii} \bP(y\in I) &= 1-o(1).
  \end{align}
\end{lemma}
\begin{proof}
  \cref{eqn:lem-th-binom-concentrate:i} is by \cref{lem:binomial-large-deviation}.
  \cref{eqn:lem-th-binom-concentrate:ii} is by Chernoff bound.
\end{proof}

Let $\bP^{(0)}$ denote the prior distribution of $A$ and $\bP^{(1)}$ denote the posterior distribution of $A$ conditioned on observations in Phase 1.
Let $\cC^{(0)}$ (resp.~$\cC^{(1)}$) denote the support of $\bP^{(0)}$ (resp.~$\bP^{(1)}$).
Then $\cC^{(1)} = \cC^{(0)} = \binom{[n]}{k-1} \cup \binom{[n]}k$ and for any set $B \in \cC^{(0)}$ we have
\begin{align*}
  \bP^{(1)}(B) &\propto \bP\left((a_i)_{i\in [n]} | B\right) \bP^{(0)}(B) \\
  \nonumber &= \left(\prod_{i\in B} p_{a_i}\right) \left(\prod_{i\in B^c} p_{m_1-a_i}\right) \bP^{(0)}(B).
\end{align*}

\subsection{Phase 2} \label{sec:th-small:phase-2}
In Phase 2, the oracle reveals some elements in $A$ and not in $A$ as follows.
\begin{enumerate}[label=2\alph*.]
  \item In Step 2a, the oracle reveals elements $i$ with $a_i\not \in I$.
  \item In Step 2b, the oracle reveals every $i\in A$ independently with probability $q_{a_i}$. We choose $q_j = 1-\frac{p_{m_1-j} p_{j_l}}{p_j p_{m_1-j_l}}$ for $j\in I$ where $j_l = p m_1 - m_1^{0.6}$.
  \item In Step 2c, the oracle reveals $k-1$ elements of $A$ as follows. If $|A|=k-1$, reveal all elements of $A$. Otherwise, $|A|=k$. If $A$ contains an element that is not revealed yet, uniformly randomly choose an element $i^*$ from all such elements and reveal all elements in $A\backslash i^*$. If all elements of $A$ have been revealed, report failure.
\end{enumerate}

\paragraph{Step 2a and Step 2b.}
By the same analysis as in \cref{sec:inf:phase-2:step-a,sec:inf:phase-2:step-b}, observations up to Step 2b have the same effect as the following procedure:
\begin{enumerate}[label=(\arabic*)]
  \item Observe every element $i\in A$ independently with probability
  \begin{align*}
    p_+ &= 1-\sum_{j\in I} p_j (1-q_j)
    = 1 - \frac{p_{j_l}}{p_{m_1-j_l}} \cdot \sum_{j\in I} p_{m_1-j} \\
    \nonumber &= 1-(1\pm o(1)) \frac{p_{j_l}}{p_{m_1-j_l}} = 1-(\delta/k)^{1-\epsilon_1\pm o(1)}.
  \end{align*}
  \item Observe every element $i\in A^c$ independently with probability
  \begin{align*}
    p_- = \bP(\Bin(m_1,p)\not \in I) = o(1).
  \end{align*}
\end{enumerate}

Let $A_+^{(2a)}$ (resp.~$A_-^{(2a)}$) denote the set of elements in $I$ (resp.~not in $I$) revealed in Step 2a.
Let $A_+^{(2b)}$ be the set of elements in $I$ revealed in Step 2b, and $A_+^{(\le 2b)} = A_+^{(2a)} \cup A_+^{(2b)}$.
Then the oracle reports failure in Step 2c if and only if $\left|A_+^{(\le 2b)}\right| = k$.

Let $\bP^{(2b)}$ denote the posterior distribution of $A$ after Step 2b and $\cC^{(2b)}$ be its support.
By the same analysis as in \cref{sec:inf:phase-2:step-a,sec:inf:phase-2:step-b}, the posterior probability after Step 2b satisfies
\begin{align*}
  \bP^{(2b)}(B) \propto \left(\frac{p_{j_l}}{p_{m_1-j_l}}\right)^{|B|-(k-1)} \bP^{(0)}(B)
\end{align*}
for $B\in \cC^{(2b)}$.
In particular, $\left( A_+^{(\le 2b)}, A_-^{(2a)}, |A|\right)$ is a sufficient statistic for $A$ at the end of Step 2b.

We note that $\left|A_+^{(\le 2b)}\right| < k$ with probability $\ge \delta^{1-\epsilon_1\pm o(1)}$ because
\begin{align} \label{eqn:th:step-2b:zero-prob}
  \bP\left( \Bin(k, 1-p_+) > 0 \right) = 1-p_+^k
  \ge 1-\exp\left( k(1-p_+)\right) \ge \delta^{1-\epsilon_1\pm o(1)}.
\end{align}

\paragraph{Step 2c.}
Let $A_+^{(2c)}$ be the set of elements in $I$ revealed in Step 2c but not in previous steps, and $A_+^{(\le 2c)} = A_+^{(\le 2b)} \cup A_+^{(2c)}$.
Let $\bP^{(2c)}$ be the posterior distribution of $A$ and $\cC^{(2c)}$ be the support of $\bP^{(2c)}$.
Then
\begin{align*}
  \cC^{(2c)}&=\left\{ A : A\in \binom{[n]}{k-1} \cup \binom{[n]}k, A_+^{(\le 2c)} \subseteq A \subseteq [n]\backslash A_-^{(2a)}\right\} \\
  \nonumber &= \left\{A_+^{(\le 2c)}\right\} \cup \left\{ A_+^{(\le 2c)} \cup \{i\}: i\in [n]\backslash \left(A_+^{(\le 2c)} \cup A_-^{(2a)}\right)\right\}.
\end{align*}
For simplicity of notation, define $A_0 = A_+^{(\le 2c)}$ and $A_i = A_+^{(\le 2c)} \cup \{i\}$ for $i\in [n]\backslash \left(A_+^{(\le 2c)} \cup A_-^{(2a)}\right)$.
Then
\begin{align*}
  \bP^{(2c)}(A_0) &\propto \bP^{(0)}(A_0),\\
  \bP^{(2c)}(A_i) &\propto \frac{1}{k-\left| A_+^{(\le 2b)} \right|} \cdot \frac{p_{j_l}}{p_{m_1-j_l}} \cdot \bP^{(0)}(A_i).
\end{align*}
Recall that $\bP^{(0)}(B) = \frac 12 \cdot \frac 1{\binom {n}{|B|}}$ for $|B|\in \{k-1,k\}$.
Let $\bP^{(2c)}_k$ (resp.~$\bP^{(2c)}_{k-1}$) denote the probability measure of observations at the end of Step 2c conditioned on $|A|=k$ (resp.~$|A|=k-1$).
Summing over $i$, we have
\begin{align*}
  \frac{d \bP^{(2c)}_k}{d \bP^{(2c)}_{k-1}} =\left| [n]\backslash \left(A_+^{(\le 2c)} \cup A_-^{(2a)}\right) \right|\cdot \frac{1}{k-\left| A_+^{(\le 2b)} \right|} \cdot \frac{p_{j_l}}{p_{m_1-j_l}} \cdot \frac{\binom n{k-1}}{\binom nk}.
\end{align*}
\begin{lemma} \label{lem:th:step-2c}
  \begin{align*}
    \bP^{(2c)}_k \left(\frac{d \bP^{(2c)}_k}{d \bP^{(2c)}_{k-1}} \le \delta^{-\epsilon_1/2}\right) \ge \delta^{1-\epsilon_1 \pm o(1)}.
  \end{align*}
\end{lemma}
\begin{proof}
  Because $p_-=o(1)$, with probability $1-o(1)$, we have $\left|A_-^{(2a)}\right| \le n p_-^{1/2}$.
  Then
  \begin{align*}
    \left| [n]\backslash \left(A_+^{(\le 2c)} \cup A_-^{(2a)}\right) \right|
    = n-(k-1) - A_-^{(2a)} = (1\pm o(1)) n
  \end{align*}
  and
  \begin{align*}
    \frac{d \bP^{(2c)}_k}{d \bP^{(2c)}_{k-1}}
    &= (1\pm o(1)) n \cdot \frac{1}{k-\left| A_+^{(\le 2b)} \right|} \cdot \frac{p_{j_l}}{p_{m_1-j_l}} \cdot \frac kn\\
    \nonumber &= (1\pm o(1)) \frac{k}{k-\left| A_+^{(\le 2b)} \right|} \cdot \frac{p_{j_l}}{p_{m_1-j_l}}.
  \end{align*}

  Note that $k-\left| A_+^{(\le 2b)} \right|\sim \Bin(k,1-p_+)$, and $\left| A_+^{(\le 2b)} \right|$ is independent with $\left|A_-^{(2a)}\right|$ conditioned on $|A|$. Recall that
  \begin{align*}
    1-p_+ = (1\pm o(1)) \frac{p_{j_l}}{p_{m_1-j_l}} = (\delta/k)^{1-\epsilon_1 \pm o(1)}.
  \end{align*}

  If $k^{\epsilon_1} \delta^{1-\epsilon_1} \ge \delta^{-\epsilon_1/2}$, then $k (1-p_+) = \omega(1)$ and by concentration
  \begin{align*}
    \bP\left( \Bin(k,1-p_+) \ge \frac 12 k(1-p_+) \right) = 1-o(1)
  \end{align*}
  Under $\bP^{(2c)}_k$, conditioned on $k-\left| A_+^{(\le 2b)} \right| \ge \frac 12 k(1-p_+)$, we have
  \begin{align*}
    \frac{d \bP^{(2c)}_k}{d \bP^{(2c)}_{k-1}}
    = (1\pm o(1)) \frac{k}{k-\left| A_+^{(\le 2b)} \right|} \cdot \frac{p_{j_l}}{p_{m_1-j_l}}
    = O(1).
  \end{align*}

  If $k^{\epsilon_1} \delta^{1-\epsilon_1} \le \delta^{-\epsilon_1/2}$, then conditioned on $\left| A_+^{(\le 2b)} \right| < k$ (which happens with probability at least $\delta^{1-\epsilon_1\pm o(1)}$ by \cref{eqn:th:step-2b:zero-prob}), we have
  \begin{align*}
    \frac{d \bP^{(2c)}_k}{d \bP^{(2c)}_{k-1}}
    = (1\pm o(1)) \frac{k}{k-\left| A_+^{(\le 2b)} \right|} \cdot \frac{p_{j_l}}{p_{m_1-j_l}}
    \le \delta^{-\epsilon_1/2 \pm o(1)}.
  \end{align*}

  Combining both cases we finish the proof.
\end{proof}
From \cref{lem:th:step-2c} we can conclude that no algorithm can determine $|A|$ with error probability $\le \delta$ at the end of Phase 2.
Suppose for the sake of contradiction that such an algorithm exists.
Let $\cE$ denote the event that $\frac{d \bP^{(2c)}_k}{d \bP^{(2c)}_{k-1}} \le \delta^{-\epsilon_1/2}$.
Under $\bP^{(2c)}_k$, conditioned on $\cE$, the algorithm outputs ``$|A|=k$'' with probability at least $\frac 12$ (otherwise the overall error probability under $\bP^{(2c)}_k$ is at least $\frac 12 \cdot \delta^{1-\epsilon\pm o(1)}$).
However, this implies that under $\bP^{(2c)}_{k-1}$, the algorithm outputs $|A|=k$ with probability at least $\frac 12 \delta^{\epsilon_1/2}$ by definition of $\cE$.

\subsection{Phase 3} \label{sec:th-small:phase-3}
In Phase 3, the algorithm makes at most $(1-\epsilon_2) n$ adaptive exact queries.
If $\left| [n]\backslash \left(A_+^{(\le 2c)} \cup A_-^{(2a)}\right) \right| = (1\pm o(1)) n$, then with probability $(1\pm o(1))\epsilon_2$, these exact queries do not hit any elements in $A$.
Let $\bP^{(3)}_k$ (resp.~$\bP^{(3)}_{k-1}$) denote the probability measure of observations at the end of Phase 3 conditioned on $|A|=k$ (resp.~$|A|=k-1$).
Conditioned on that the exact queries do not hit any elements in $A$, we have
\begin{align*}
  \frac{d \bP^{(2c)}_k}{d \bP^{(2c)}_{k-1}} =\left| [n]\backslash \left(A_+^{(\le 2c)} \cup A_-^{(2a)} \cup A^{(3)}\right) \right|\cdot \frac{1}{k-\left| A_+^{(\le 2b)} \right|} \cdot \frac{p_{j_l}}{p_{m_1-j_l}} \cdot \frac{\binom n{k-1}}{\binom nk}
\end{align*}
where $A^{(3)}$ denotes the set of elements queried in Phase 3.
By a similar proof as \cref{lem:th:step-2c}, we can prove that
\begin{align*}
  \bP^{(3)}_k \left(\frac{d \bP^{(3)}_k}{d \bP^{(3)}_{k-1}} \le \epsilon_2^{-1}\delta^{-\epsilon_1/2}\right) \ge \epsilon_2 \delta^{1-\epsilon_1 \pm o(1)}.
\end{align*}
However, $\epsilon_2$ is a constant, so the discussion in the end of \cref{sec:th-small:phase-2} still applies.
This concludes that no algorithm can solve the three-phase problem with error probability $\le \delta$ using at most $(1-\epsilon) \frac{n \log \frac k\delta}{\DKL}$ noisy queries.

\end{document}